\documentclass[10pt,a4paper]{article}
\usepackage[utf8]{inputenc}
\usepackage[T1]{fontenc}

\newif\iflong
\newif\ifshort
\longtrue
\iflong
\else
\shorttrue
\fi

\newcommand{\bigoh}{\mathcal{O}}

\usepackage{fullpage}
\usepackage{complexity}
\usepackage{todonotes}
\usepackage{boxedminipage,hyperref}
\usepackage[most]{tcolorbox}
\usepackage{xspace}
\usepackage{wasysym}
\usepackage{graphicx}
\usepackage{subcaption}
\usepackage{amsmath,amssymb,amsthm}
\usepackage{cleveref}
\usepackage{authblk}

\usepackage{url,hyperref}
\hypersetup{
	hidelinks,
	colorlinks=true,
	citecolor=[rgb]{0.121 0.47 0.705},
	linkcolor=[rgb]{0.121 0.47 0.705},
	urlcolor=[rgb]{0.121 0.47 0.705}
}

\DeclareMathOperator{\interior}{int}

\makeatletter
\g@addto@macro\bfseries{\boldmath}
\makeatother

\newbox{\myorcidauthbox}
\sbox{\myorcidauthbox}{\large\includegraphics[height=1.7ex]{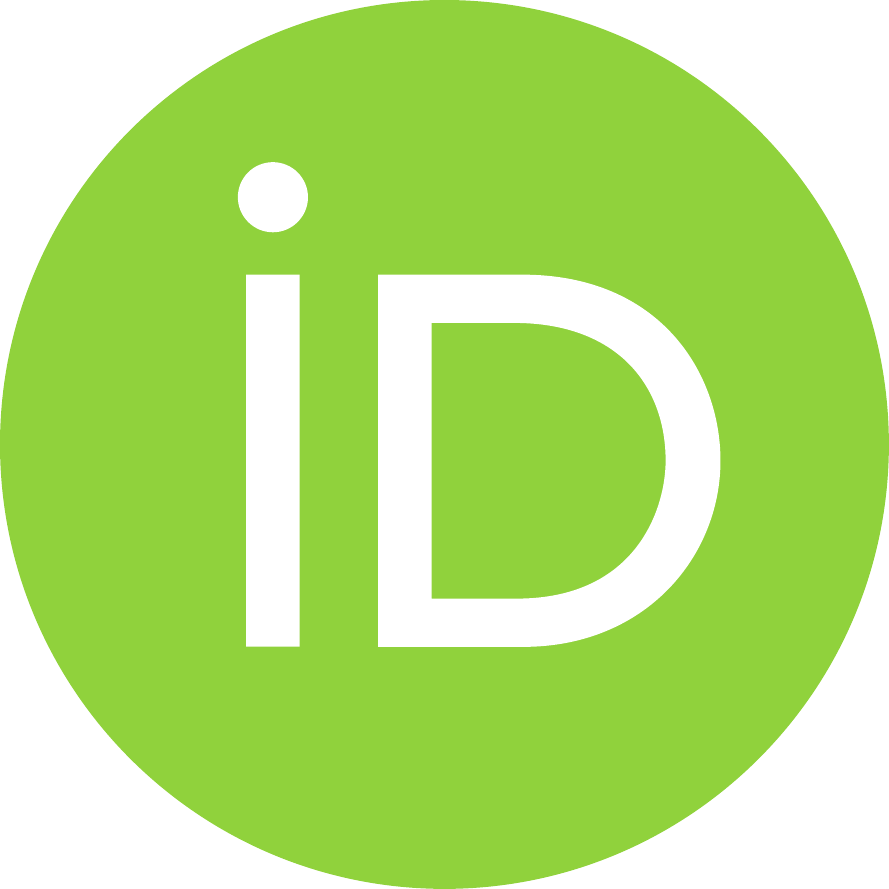}}
\newcommand{\orcid}[1]{%
	\href{https://orcid.org/#1}{\usebox{\myorcidauthbox}}}

\author[1]{Robert Ganian%
	\orcid{0000-0002-7762-8045}%
	\thanks{%
		Supported by the Austrian Science Fund (FWF) via project P31336 (\emph{New Frontiers for Parameterized Complexity}).}}
\author[1]{Thekla Hamm\protect\footnotemark[2]}
\author[2]{Fabian Klute%
	\orcid{0000-0002-7791-3604}%
	\thanks{%
		Supported by the Netherlands Organisation for Scientific Research (NWO) under project no. 612.001.651.}}
\author[3]{Irene Parada\orcid{0000-0003-3147-0083}}
\author[4]{Birgit Vogtenhuber%
	\orcid{0000-0002-7166-4467}%
	\thanks{%
		Partially supported by Austrian Science Fund (FWF) within the collaborative DACH project \emph{Arrangements and Drawings} as FWF project \mbox{I 3340-N35}.
}}

\affil[1]{Algorithms and Complexity Group, TU Wien, Austria}
\affil[2]{Utrecht University, The Netherlands}
\affil[3]{TU Eindhoven, The Netherlands}
\affil[4]{Graz University of Technology, Austria}
\affil[ ]{
	
	\texttt{
	[rganian|thamm]@ac.tuwien.ac.at\\
	f.m.klute@uu.nl\\
	i.m.de.parada.munoz@tue.nl\\
	bvogt@ist.tugraz.at}
}

\date{}

\title{Crossing-Optimal Extension of Simple Drawings\footnote{
		This work was started during the 
		Austrian Computational Geometry Reunion Meeting, August 10 to 14, in Strobl, Austria. We thank all the participants for the nice working atmosphere as well as fruitful discussions on this as well as other topics. 
		The authors would also like to thank Eduard Eiben for his insightful comments.
}}
\newtheorem{observation}{Observation}

\newtheorem{newcounterclaim}{Claim}
\newtheorem{theorem}{Theorem}
\newtheorem{lemma}[theorem]{Lemma}
\newtheorem{proposition}[theorem]{Proposition}
\newtheorem{definition}[theorem]{Definition}
\newtheorem{corollary}[theorem]{Corollary}

\newcommand{\patchwork}{\ensuremath{P}\xspace}

\newcommand{\drawing}{\ensuremath{\mathcal G}\xspace}
\newcommand{\planarization}[1]{\ensuremath{#1^\times}\xspace}
\newcommand{\dual}[1]{\ensuremath{#1^*}\xspace}
\newcommand{\pgraph}{\planarization{\drawing}}
\newcommand{\pdrawing}{\planarization{\drawing}}
\newcommand{\dgraph}{\dual{G}}
\newcommand{\edrawing}{\ensuremath{\mathcal G^+}\xspace}
\newcommand{\egraph}{\ensuremath{G^+}\xspace}

\newcommand{\newedges}{\ensuremath{F}\xspace}

\newcommand{\SDE}{\(\mathfrak{Z}\)\xspace}

\newcommand{\SCEI}{\textsc{SCEI}\xspace}
\newcommand{\SLCEI}{\textsc{SLCEI}\xspace}
\newcommand{\SConeEI}{\textsc{SC1EI}\xspace}

\newcommand{\hole}{\ensuremath{H}\xspace}

\newcommand{\xrep}[1]{\ensuremath{\subseteq^{#1}_{\text{rep}}}\xspace}
\newcommand{\qrep}{\xrep{q}}
\newclass{\ETH}{ETH}

\usepackage{apptools}

\newcounter{section-preserve}
\newcounter{theorem-preserve}
\newcounter{figure-preserve}
\newcommand{\blank}[1]{}
\newcommand{\short}[1]{\ifbool{long}{}{#1}}
\newcommand{\longonly}[1]{\ifbool{long}{#1}{}}
\newtoks\magicAppendix
\magicAppendix={}
\newtoks\magictoks
\newif\iflater
\laterfalse
\ifbool{long}{%
	\newcommand{\later}[1]{#1}%
}%
{%
\long\def\later#1{\magictoks={#1}%
	\edef\magictodo{\noexpand\magicAppendix={\the\magicAppendix%
			\the\magictoks%
	}}%
	\magictodo}%
}%
\long\def\latertitle#1{\magictoks={#1}%
	\edef\magictodo{\noexpand\magicAppendix={\the\magicAppendix \par
			\the\magictoks%
	}}%
	\magictodo}

\long\def\both#1{\magictoks={#1}%
	\edef\magictodo{\noexpand\magicAppendix={\the\magicAppendix%
			\noexpand\setcounter{theorem-preserve}{\noexpand\arabic{theorem}}%
			\noexpand\setcounter{theorem}{\arabic{theorem}}%
			\noexpand\setcounter{section-preserve}{\noexpand\arabic{section}}%
			\noexpand\setcounter{section}{\arabic{section}}%
			\noexpand\setcounter{figure-preserve}{\noexpand\arabic{figure}}%
			\noexpand\setcounter{figure}{\arabic{figure}}%
			\noexpand\let\noexpand\oldsection=\noexpand\thesection%
			\noexpand\def\noexpand\thesection{\thesection}%
			\noexpand\let\noexpand\oldlabel=\noexpand\label%
			\noexpand\let\noexpand\label=\noexpand\blank%
			\the\magictoks%
			\noexpand\setcounter{theorem}{\noexpand\arabic{theorem-preserve}}%
			\noexpand\setcounter{section}{\noexpand\arabic{section-preserve}}%
			\noexpand\setcounter{figure}{\noexpand\arabic{figure-preserve}}%
			\noexpand\let\noexpand\thesection=\noexpand\oldsection%
			\noexpand\let\noexpand\label=\noexpand\oldlabel%
	}}%
	\magictodo%
	\the\magictoks}%
\def\magicappendix{\latertrue \the\magicAppendix}

\makeatletter
\newcommand{\appendixandlong}[1]{%
	\ifbool{long}{%
		#1%
	}{%
		\IfAppendix{%
			#1%
		}%
		{}%
	}%
}%
\makeatother

\makeatletter
\newcommand{\appendixandlongorshort}[2]{%
	\ifbool{long}{%
		#1%
	}{%
		\IfAppendix{%
			#1%
		}%
		{%
			#2%
		}%
	}%
}%
\makeatother

\makeatletter
\patchcmd{\thmhead}{\thmnote{ {\the\thm@notefont(#3)}}}%
{%
	\ifbool{long}{%
		\ifstrequal{#3}{$\star$}{%
		}{%
			\thmnote{ {\the\thm@notefont(#3)}}%
		}%
	}{%
		\IfAppendix{%
			\ifstrequal{#3}{$\star$}{%
			}{%
				\thmnote{ {\the\thm@notefont(#3)}}%
			}%
		}%
		{%
			\thmnote{ {\the\thm@notefont(#3)}}%
		}%
	}%
}
{}{} 
\makeatother

\begin{document}
	\maketitle
	
	\begin{abstract}
In extension problems of partial graph drawings one is given an incomplete drawing of an input graph $G$ and is asked to complete the drawing while maintaining certain properties. A prominent area where such problems arise is that of crossing minimization. For plane drawings and various relaxations of these, there is a number of tractability as well as lower-bound results exploring the computational complexity of crossing-sensitive drawing extension problems. In contrast, comparatively few results are known on extension problems for the fundamental and broad class of simple drawings, that is, drawings in which each pair of edges intersects in at most one point. 
\ifshort		
In fact, the extension problem of simple drawings has only recently been shown to be \NP-hard even for inserting a single edge.
\fi
\iflong
In fact, only recently it has been shown that the extension problem of simple drawings is \NP-hard even when the task is to insert a single edge.
\fi

In this paper we present tractability results for the crossing-sensitive extension problem of simple drawings. In particular, we show that the problem of inserting edges into a simple drawing is fixed-parameter tractable when parameterized by the number of edges to insert and an upper bound on newly created crossings. Using the same proof techniques, we are also able to answer several closely related variants of this problem, among others the extension problem for $k$-plane drawings. Moreover, using a different approach, we provide a single-exponential fixed-parameter algorithm for the case in which we are only trying to insert a single edge into the drawing.
%
	\end{abstract}
	
	\section{Introduction}
	\label{sec:intro}
	The study of the crossing number of graphs, that is, the minimum number of edge crossings necessary to draw a given graph, is a major research direction in the field of computational geometry~\cite{DBLP:reference/crc/BuchheimCGJM13,pachhandbook17,schaefer2018crossing}. 
	More recently, there have been a number of works focusing on minimizing or restricting edge crossings when the task is not to draw a graph from scratch, but rather to extend a partial drawing that is provided on the input.
	Prominently, Chimani et al.~\cite{ChimaniGMW09} showed that extending a plane drawing with a star in a way that minimizes the number of crossings of the resulting drawing is polynomial-time tractable. Later, Angelini et al.~\cite{AngeliniBFJKPR15} obtained a polynomial-time algorithm for extending plane drawings so that the crossing number remains $0$ (i.e., the resulting drawing is plane). 

	
While the two results mentioned above give rise to polynomial-time variants of 
crossing-minimization extension problems, a number of important cases are known to be \NP-hard; 
a prototypical example is the \textsc{Rigid Multiple Edge Insertion (RMEI)} problem, which asks for a crossing-minimal insertion of $k$ edges into a plane drawing of an $n$-vertex graph~\cite{ChimaniH16,Ziegler01}. 
To deal with this, in recent years the focus has broadened to also consider a weaker notion of tractability, namely, \emph{fixed-parameter tractability} (\FPT)~\cite{CyganFKLMPPS15,DowneyF13}. 
Chimani and Hlin\v en\' y~\cite{ChimaniH16} have shown that \textsc{RMEI} 
is \FPT, i.e., there is an algorithm which solves that problem in time $f(k)\cdot n^{\bigoh(1)}$. 
Other works have considered various relaxations of crossing minimization; 
for instance, recently Eiben et al.~\cite{EibenGHKN20} established the fixed-parameter tractability of extending drawings by $k$ edges in a way which does not minimize the total number of crossings, but rather bounds the number of crossings per edge to at most $1$.

\short{%
For many problems in the intersection of crossing minimization and graph extension, an important goal is that the desired extension should maintain certain properties of the given partial representation. 
In the problems studied in \cite{AngeliniBFJKPR15} and \cite{EibenGHKN20}, the input is a plane or 1-plane\footnote{A drawing of a graph is $\ell$-plane if every edge is involved in at most $\ell$ crossings.} drawing, respectively, and 
the desired extension must maintain the property of being (1-)plane.
There have been a plethora of results exploring such extension problems, especially on plane drawings, for a range of other, often more restrictive properties~\cite{angeliniExtendingPartialOrthogonal2020,PartialConstrainedLevel_Brueckner_2017,cegl-dgpwpofpa-12,cfglms-dpespg-15,lozzo_compgeo_2020,ExtendingConvexPartial_Mchedlidze_2015,ext_straight_06}.}

\longonly{%
It is worth noting that for many questions in the intersection of minimizing the number of crossings and the study of graph extension problems, an important goal is that the desired extension should maintain certain properties of the given partial representation. 
For instance, in the problems studied by Angelini et al.~\cite{AngeliniBFJKPR15} and Eiben et al.~\cite{EibenGHKN20}, the input is a plane or 1-plane\footnote{A drawing of a graph is $\ell$-plane if every edge is involved in at most $\ell$ crossings.} drawing, respectively, and 
the desired extension must maintain the property of being (1-)plane.
In a similar manner, recent years have seen a wide range of results exploring the algorithmic complexity of such extension problems maintaining more restricted properties.
For plane drawings, classical restrictions of the drawing such as it being straight-line~\cite{ext_straight_06},
level-planar~\cite{PartialConstrainedLevel_Brueckner_2017},
upward~\cite{lozzo_compgeo_2020}, or
orthogonal~\cite{angeliniExtendingPartialOrthogonal2020}
have been explored.
Other results for planar graphs consider the number of required bends~\cite{cfglms-dpespg-15} or 
assume that the partially drawn subgraph is a cycle~\cite{cegl-dgpwpofpa-12,ExtendingConvexPartial_Mchedlidze_2015}.
For $1$-plane drawings, the same authors as in~\cite{EibenGHKN20} very recently extended their results~\cite{eiben_mfcs_2020}.}

\short{%
Beyond planarity, the perhaps most prominent class of drawings with respect to crossing minimization are \emph{simple drawings}
(also called  \emph{good drawings}~\cite{AMRS18,EG_1973}, \emph{simple topological graphs}~\cite{kyncl2009}, or simply \emph{drawings}~\cite{h-etdcg-98}).
A drawing is simple if every pair of edges intersect in at most one point that is either a common endpoint or a proper crossing.
Simplicity is an extremely natural restriction that is taken as a basic assumption in a range of settings, e.g.,~\cite{DBLP:journals/jctb/AngeliniBBLBDHL20,extending_pseudolinear,DBLP:journals/jocg/CardinalF18,DBLP:journals/dcg/Kyncl20},
and that constitutes a necessary requirement for crossing-minimal drawings~\cite{schaefer2018crossing}.}

\longonly{%
Leaving planar and just beyond planar graphs,
the most prominent class of drawings with respect to crossing minimization are \emph{simple drawings}
(also called  \emph{good drawings}~\cite{AMRS18,EG_1973}, \emph{simple topological graphs}~\cite{kyncl2009}, or simply \emph{drawings}~\cite{h-etdcg-98}).
A drawing is simple if every pair of edges intersect in at most one point that is either a common endpoint or a proper crossing.
Simplicity is an extremely natural restriction that is 
taken 
as a basic assumption in a range of settings, e.g.,~\cite{DBLP:journals/jctb/AngeliniBBLBDHL20,extending_pseudolinear,DBLP:journals/jocg/CardinalF18,DBLP:journals/dcg/Kyncl20}, and
it is known that simplicity is a necessary requirement for crossing-minimal drawings~\cite{schaefer2018crossing}.
}

\smallskip \noindent \textbf{Contribution.} \quad 
In this work we study the extension problem for simple drawings in the context of crossing minimization.
 In other words, our aim is to extend a given simple drawing with $k$ new edges while maintaining simplicity and restricting newly created crossings. Naturally, the most obvious way of restricting such crossings is by bounding their number, leading us to our first problem of interest:\footnote{\emph{Decision} versions of problems are provided purely for complexity-theoretic reasons; every algorithm provided in this article is constructive and can also output a solution as a witness.}

\begin{center}
\begin{boxedminipage}{0.98 \columnwidth}
\textsc{Simple Crossing-Minimal Edge Insertion} (\SCEI)\\[2pt]
\begin{tabular}{l p{0.80 \columnwidth}}
Input: & A graph $G=(V,E)$ along with a connected simple drawing $\drawing$, an integer $\ell$, and a set $F$ of $k$ edges of the complement of $G$.\\
Question: & Can $\drawing$ be extended to a simple drawing $\drawing'$ of the graph $G'=(V,E\cup F)$ such that the number of crossings in $\drawing'$ that involve an edge of $F$ is at most $\ell$?
\end{tabular}
\end{boxedminipage}
\end{center}
Note that we require the initial drawing \drawing to be connected.
While this is a natural assumption that is well-justified in many situations, it would certainly also make sense to consider the more general setting in which this is not the case. A short discussion of how the connectivity of \drawing\ is used in our proof is provided in Section~\ref{sec:patchworkdef}.


\SCEI\ was recently shown to be \NP-complete already when $|F|=1$ and $\ell\geq|E|$ (meaning that the aim is merely to obtain a simple drawing)~\cite{ArroyoKPSVW20}.
\short{On the other hand, dropping the simplicity requirement of the resulting drawing, the problem reduces to \textsc{RMEI} which is \FPT.}%
\longonly{%
On the other hand, if one drops the requirement that the resulting drawing must be simple, the problem immediately reduces to the \textsc{Rigid Multiple Edge Insertion} problem that was shown to be \FPT\ (parameterized by the number of inserted edges) by Chimani and Hlin\v en\' y~\cite{ChimaniH16}---indeed, one can simply planarize the initial drawing and then apply the provided algorithm.
}

\short{
The main contribution of this article is an \FPT\ algorithm for \SCEI\ parameterized by $k+\ell$. The result is obtained via a combination of techniques recently introduced 
in~\cite{EibenGHKN20} and completely new machinery. 
A high-level overview of challenges posed by the problem and our strategies for overcoming them is provided in the next part of this introduction. Before that, let us mention other natural crossing-sensitive restrictions of simple drawing extension.
}
\longonly{
The main contribution of this article is a fixed-parameter algorithm for \SCEI\ parameterized by $k+\ell$. The result is obtained via a combination of the techniques recently introduced by Eiben et al.~\cite{EibenGHKN20} and completely new machinery. 
A high-level overview of the challenges posed by the problem and our strategy for overcoming these challenges is provided in the next part of this introduction. However, before proceeding there, let us mention other natural crossing-sensitive restrictions of simple drawing extension.
}

Instead of restricting the \emph{total number} of newly created crossings, one may aim to extend~$\drawing$ in a way which bounds the number of crossings involving each of the newly added edges---akin to the restrictions imposed by $\ell$-planarity. We call this problem \textsc{Simple Locally Crossing-Minimal Edge Insertion} (\textsc{SLCEI}), where the role of $\ell$ is that it bounds the maximum number of crossings involving any one particular edge of $F$. Alternatively, one may simply require that \emph{every} edge in the resulting drawing is involved in at most $\ell$ crossings, i.e., that the whole $\drawing'$ is $\ell$-plane. This results in the \textsc{Simple $\ell$-Plane Edge Insertion} (\textsc{S$\ell$-PEI}) problem. Both of these problems are known to be \NP-hard when either $\ell=1$ or $k=1$, meaning that we can drop neither of our parameters if we wish to achieve tractability.

One key strength
of the framework we develop for solving \SCEI\ is its universality. Notably, we obtain the fixed-parameter tractability of \textsc{SLCEI} as an immediate corollary of the proof of our main theorem, while the fixed-parameter tractability of \textsc{S$\ell$-PEI} follows by a minor adjustment of the final part of our proof. 
Moreover, it is trivial to use the framework to solve the considered problems when one drops the requirement that the final drawing is simple---allowing us to, e.g., generalize the previously established fixed-parameter tractability of \textsc{1-Planar Edge Insertion}~\cite{EibenGHKN20} to \textsc{$\ell$-Planar Edge Insertion} ($\ell$-\textsc{PEI}).

Finally, we note that a core ingredient in our approach is the use of Courcelle's theorem~\cite{Courcelle90}, and hence the algorithms underlying our tractability results will have an impractical dependency on $k$. However, for the special case of $|F|=1$ (i.e., when inserting a single edge), we use so-called \emph{representative sets} to provide a single-exponential fixed-parameter algorithm which is tight under the exponential time hypothesis~\cite{ImpagliazzoPZ01}.
	
\smallskip \noindent \textbf{Proof Overview.} \quad
\short{%
On a high level, our approach 
follows the general strategy co-developed by a subset of the authors in~\cite{EibenGHKN20} 
for solving the problem of inserting $k$ edges into a drawing while maintaining $1$-planarity. This general strategy can be summarized as follows:}
\longonly{%
On a high level, our approach for establishing the fixed-parameter tractability of the considered problems follows the general strategy co-developed by a subset of the authors to solve \textsc{1-Planar Edge Drawing Extension} (\textsc{1PEDE}), i.e., the problem of inserting $k$ edges into a drawing while maintaining $1$-planarity. This general strategy can be summarized as follows:
}

%
	\begin{enumerate}
	\item We preprocess $G$ and a planarization of $\drawing$ to remove parts of \drawing\ which are too far away to interact with our solution. This drawing is then translated into a graph representation of bounded \emph{treewidth}~\cite{RobertsonS84}.
	\item We identify a combinatorial characterization that captures how the solution curves will be embedded into $\drawing$. Crucially, the characterization has size bounded by our parameters.
	\item We perform brute-force branching over all characterizations to pre-determine the behavior of a solution in $\drawing$, and for each such characterization we employ \emph{Courcelle's theorem}~\cite{Courcelle90} to determine whether there exists a solution with such a characterization.
	\end{enumerate}
	
The specific implementation of this strategy differs substantially from the previous work~\cite{EibenGHKN20}---for instance, the combinatorial characterization of solutions in Step 2 and the use of Courcelle's theorem in Step 3 are both different. But the by far greatest challenge in implementing this strategy occurs in Step 1. Notably, removing the parts of $\drawing$ required to obtain a bounded-treewidth graph representation creates \emph{holes} in the drawing, and these could disconnect edges intersecting these holes. The graph representation can then lose track of ``which edge parts belong to each other'', which means we can no longer use it to determine whether the extended drawing is simple. We remark that specifically for \textsc{S$\ell$-PEI} and \textsc{$\ell$-PEI}, it would be possible to directly adapt Step 1 to ensure that no edge is disconnected in this manner, thus circumventing this difficulty.


To handle this problem, we employ an in-depth geometric analysis combined with a careful use of the sunflower lemma and subroutines which invoke Courcelle's theorem to construct a representation which (a) still has bounded treewidth, and (b) contains partial information about which edge parts belong to the same edge in $\drawing$. 
\short{%
A detailed overview of how this is achieved is presented at the beginning of Section~\ref{sec:patchwork}.
}
\longonly{%
A detailed overview of how this is achieved is presented at the beginning of Section~\ref{sec:patchwork}---but for a high-level and simplified intuition, let us imagine that we wish to extend $\drawing$ with an $s$-$t$ curve $\pi$ while maintaining simplicity and achieving only few crossings, and we are worried about long-distance edges in $\drawing$ which may be intersected by $\pi$ but also contain parts that are very ``far'' from $s$ and $t$. We essentially show that there is an efficiently computable and small set of special long-distance edges in $\drawing$ that is sufficient to find one solution, and all other long-distance edges may be disregarded (even though there may exist some other solutions using these).}

\smallskip

\noindent \textbf{Related Work.} \quad
\longonly{%
There have been two distinct lines of work that recently considered simple drawings in the context of drawing extension problems, albeit with different goals. 
First, Hajnal et al.~\cite{hajnal2015saturated} and Kyn{\v{c}}l et al.~\cite{DBLP:journals/comgeo/KynclPRT15} studied \emph{saturated} simple drawings,  i.e., simple drawings in which no edge can be inserted without violating simplicity.
Second, a number of authors have studied 
the computational complexity of deciding whether a given simple drawing can be extended with a given set of edges while maintaining simplicity~\cite{adp-esd-19,ArroyoKPSVW20}. 
}
\short{%
There have been two distinct lines of work that recently considered simple drawings in the context of drawing extension problems. The first studied a closely related notion of \emph{saturated} simple drawings~\cite{hajnal2015saturated,DBLP:journals/comgeo/KynclPRT15}, while the second 
studied the computational complexity of 
the extension problem for simple drawings~\cite{adp-esd-19,ArroyoKPSVW20}.
}

\short{
	
\smallskip
\noindent \emph{Statements where proofs or more details are provided in the appendix are marked with $\star$.}
}

	\section{Preliminaries}
	\label{sec:prelims}
	\latertitle{\section{Extended Version of Section~\ref{sec:prelims}}}
	\both{
	We use standard terminology for undirected and simple graphs~\cite{Diestel}. 
	The \emph{length} of a walk and path is the number of edges it visits.
	For $r \in \mathbb{N}$, we write $[r]$ as shorthand for the set $\{1,\ldots,r\}$.}
	\short{%
	 We assume a basic understanding of \emph{parameterized complexity}~\cite{CyganFKLMPPS15,DowneyF13,FlumGrohe06}, the \emph{sunflower lemma}~\cite{Erdos60,FlumGrohe06}, and \emph{Courcelle's theorem along with MSO logic}~\cite{ArnborgLS91,Courcelle90}. ($\star$)}
 
 	\both{
 		
 		\smallskip
 	
 	}
	
	\both{A \emph{simple drawing} of a graph $G$ }%
	\later{(also known as \emph{good drawing} or as \emph{simple topological graph} in the literature) }%
	\both{is a drawing \drawing of $G$ in the plane such that every pair of edges shares at most one point that
	is either a proper crossing or a common endpoint. 
	In particular, no tangencies between edges are allowed, edges must not contain any vertices in their relative interior, and
	no three edges intersect in the same point.
	Given a simple drawing \drawing of a graph $G$ and a set of edges \newedges 
	of the complement of $ G $ 
	we say that the edges in \newedges can be \emph{inserted} into \drawing if 
	there exists a simple drawing \edrawing of $\egraph=(V(G),E(G)\cup F)$ 
	that contains \drawing as a subdrawing. }%
	\later{A drawing \drawing of a graph in which each edge is crossed at most $k$ times 
	is called \emph{$k$-plane}.}%
	\both{The \emph{planarization} of a simple drawing \drawing of $G$ is the
	plane graph \pgraph obtained from \drawing by subdividing the edges of \(G\) at the crossing points of \drawing. We call each part of the subdivision of \(e \in E(G)\) in \(\pdrawing\) an \emph{edge segment} (of \(e\)).

	}
	
	\later{%
	\smallskip
	\noindent \textbf{Sunflower Lemma.}\quad 
	One tool we use to obtain our results is the classical sunflower lemma of Erd\H{o}s and Rado. A \emph{sunflower} in a set family $\mathcal{F}$ is a subset $\mathcal{F}' \subseteq \mathcal{F}$ such that all pairs of elements in $\mathcal{F}'$ have the same intersection.
%

\begin{lemma}[\cite{Erdos60,FlumGrohe06}]\label{lem:SF}
  Let $\mathcal{F}$ be a family of subsets of a universe $U$, each of cardinality exactly
  $b$, and let $a \in \mathbb{N}$. If $|\mathcal{F}|\geq b!(a-1)^{b}$, then $\mathcal{F}$
  contains a sunflower $\mathcal{F}'$ of cardinality at least $a$. Moreover,
  $\mathcal{F}'$ can be computed in time polynomial in $|\mathcal{F}|$.
\end{lemma}

	\noindent \textbf{Parameterized Complexity.}  \quad	
	In parameterized complexity~\cite{FlumGrohe06,DowneyF13,CyganFKLMPPS15},
the complexity of a problem is studied not only with respect to the input size, but also with respect to some problem parameter(s). The core idea behind parameterized complexity is that the combinatorial explosion resulting from the \NP-hardness of a problem can sometimes be confined to certain structural parameters that are small in practical settings. We now proceed to the formal definitions.

A {\it parameterized problem} $Q$ is a subset of $\Omega^* \times
\mathbb{N}$, where $\Omega$ is a fixed alphabet. Each instance of $Q$ is a pair $(I, \kappa)$, where $\kappa \in \mathbb{N}$ is called the {\it
parameter}. A parameterized problem $Q$ is
{\it fixed-parameter tractable} (\FPT)~\cite{FlumGrohe06,DowneyF13,CyganFKLMPPS15}, if there is an
algorithm, called an {\em \FPT-algorithm},  that decides whether an input $(I, \kappa)$
is a member of $Q$ in time $f(\kappa) \cdot |I|^{\bigoh(1)}$, where $f$ is a computable function and $|I|$ is the input instance size.  The class \FPT{} denotes the class of all fixed-parameter tractable parameterized problems.
A parameterized problem $Q$
is {\it \FPT-reducible} to a parameterized problem $Q'$ if there is
an algorithm, called an \emph{\FPT-reduction}, that transforms each instance $(I, \kappa)$ of $Q$
into an instance $(I', \kappa')$ of
$Q'$ in time $f(\kappa)\cdot |I|^{\bigoh(1)}$, such that $\kappa' \leq g(\kappa)$ and $(I, \kappa) \in Q$ if and
only if $(I', \kappa') \in Q'$, where $f$ and $g$ are computable
functions.

	\smallskip
	\noindent \textbf{Monadic Second Order Logic.}  \quad
	We consider \emph{Monadic Second Order} (MSO) logic on (edge-)labeled
	directed graphs in
	terms of their incidence structure, whose universe contains vertices and
	edges; the incidence between vertices and edges is represented by a
	binary relation. We assume an infinite supply of \emph{individual
		variables} $x,x_1,x_2,\dots$ and of \emph{set variables}
	$X,X_1,X_2,\dots$. The \emph{atomic formulas} are 
	$V x$ (``$x$ is a vertex''), $E y$ (``$y$ is an edge''), $I xy$ (``vertex $x$
	is incident with edge $y$''), $x=y$ (equality),
	$P_a x$ (``vertex or edge $x$ has label $a$''), and $X x$ (``vertex or
	edge $x$ is an element of set $X$'').  \emph{MSO formulas} are built up
	from atomic formulas using the usual Boolean connectives
	$(\lnot,\land,\lor,\rightarrow,\leftrightarrow)$, quantification over
	individual variables ($\forall x$, $\exists x$), and quantification over
	set variables ($\forall X$, $\exists X$).

	\emph{Free and bound variables} of a formula are defined in the usual way. 
	To indicate that the set of free individual variables of formula $\Phi$
	is $\{x_1, \dots, x_\ell\}$ and the set of free set variables of formula $\Phi$
	is $\{X_1, \dots, X_q\}$ we write $\Phi(x_1,\ldots, x_\ell, X_1,
	\dots, X_q)$. If $G$ is a graph, $v_1,\ldots, v_\ell\in V(G)\cup E(G)$ and $S_1, \dots, S_q
	\subseteq V(G)\cup E(G)$ we write $G \models \Phi(v_1,\ldots, v_\ell, S_1, \dots, S_q)$ to denote that
	$\Phi$ holds in $G$ if the variables $x_i$ are interpreted by the vertices or edges $v_i$, for $i\in [\ell]$, and the variables $X_i$ are interpreted by the sets
	$S_i$, for $i \in [q]$. 
	
	
	The following result (the well-known Courcelle's theorem~\cite{Courcelle90}) 
	shows that if $G$ has bounded treewidth~\cite{RobertsonS84} then we
	can find an assignment $\varphi$ to the set of free variables $\mathcal{F}$ with $G \models \Phi(\varphi(\mathcal{F}))$ (if one exists) in linear time. 
	
	\begin{theorem}[Courcelle's theorem~\cite{Courcelle90,ArnborgLS91}]\label{fact:MSO} 
		Let $\Phi(x_1,\dots,x_\ell, X_1,\dots, X_q)$ be a fixed MSO formula with free individual variables $x_1,\dots,x_\ell$ and free set variables $X_1,\dots,X_\ell$, and let $w$ a
		constant. Then there is a linear-time algorithm that, given a labeled
		directed graph $G$ of treewidth at most $w$, 
		either outputs  $v_1,\ldots, v_\ell\in V(G)\cup E(G)$ and $S_1, \dots, S_q	\subseteq V(G)\cup E(G)$ such that $G \models \Phi(v_1,\ldots, v_\ell, S_1, \dots, S_q)$ or correctly identifies that no such vertices $v_1,\ldots, v_\ell$ and sets $S_1, \dots, S_q$ exist.
	\end{theorem}
	
	We remark that since an understanding of the definition of \emph{treewidth} is not required for our presentation, we merely refer to the literature for a discussion of the notion~\cite{RobertsonS84,CyganFKLMPPS15,DowneyF13}. We denote the treewidth of a graph $G$ as $\operatorname{tw}(G)$.
	}

\both{%
\smallskip

\noindent \textbf{Problem Definition and Terminology.} \quad
	We formulate the following generalization of \SLCEI in which we allow the numbers of crossings allowed for each newly added edge to differ.
	Note that this formulation also fixes a parameterization.
	\begin{center}
		\begin{boxedminipage}{0.98 \columnwidth}
			\textsc{Simple Crossing-Restricted Edge Insertion (\SDE)} \hfill Parameter: $k + \max_{i \in [k]}\ell_i$\\[2pt]
			\begin{tabular}{l p{0.80 \columnwidth}}
				Input: & A graph $G=(V,E)$ along with a connected simple drawing $\drawing$, a set $\newedges = \{e_1, \dotsc, e_k\}$ of $k$ edges of the complement of $G$, and \(\ell_1, \dotsc, \ell_{k} \in \mathbb{N}\).\\
				Question: & Can $\drawing$ be extended to a simple drawing \(\drawing'\) of the graph $G'=(V,E\cup F)$ such that the drawing of each edge $e_i\in F$ has at most $\ell_i$ crossings in \(\drawing'\)?
			\end{tabular}
		\end{boxedminipage}
	\end{center}
	For an instance of \SDE we refer to elements in \(F\) as \emph{added edges} and denote the endpoints of \(e_i\) as \(s_i\) and \(t_i\) (where \(s_1,t_1 \dotsc, s_k,t_k\) are not necessarily distinct).
	For brevity we denote \(\ell = \max_{i \in [k]}\ell_i\).
	Although \SDE\ is stated as a decision problem, we will want to speak about hypothetical \emph{solutions} of \SDE, which will naturally correspond to the drawings of added edges in \(\drawing'\) (if one exists) as the rest of \(\drawing'\) is predetermined by \drawing.
	This means that a solution is a set of drawings of added edges in \(\drawing'\) where \(\drawing'\) witnesses the fact that the given instance is a \texttt{yes}-instance.
	If no such \(\drawing'\) exists, then we say that the \SDE-instance \emph{has no solution}.
	
	The reason we focus our presentation on \SDE\ is that 
	the fixed-parameter tractability of \SDE\ immediately implies the fixed-parameter tractability of both \SLCEI parameterized by the number of added edges and crossings per added edge, and \SCEI parameterized by the number of added edges and crossings of all added edges (the former problem is just a subcase of \SDE, while the latter admits a trivial \FPT-reduction to \SDE\ which simply requires that we branch to decide on an upper-bound for how many crossings each edge in $F$ contributes). Hence obtaining a fixed-parameter algorithm for \SDE\ provides a unified reason for the fixed-parameter tractability of both \SCEI\ and \SLCEI. Furthermore, we will later show that the result for \SDE\ can be straightforwardly adapted to solve the other problems mentioned in the introduction.

}
	
	\section{Stitches}
	\label{sec:patchwork}
	\latertitle{\section{Extended Version of Section~\ref{sec:patchwork}}}
	\short{%
	Let \(\left(G,\drawing,\newedges,(\ell_i)_{i \in [|\newedges|]}\right)\) be an instance of \SDE. Recalling the Proof Overview provided in Section~\ref{sec:intro}, we want to identify parts of \drawing that may be considered `unimportant' because they can never be intersected by the drawing of any of the edges \(s_it_i \in \newedges\) with at most \(\ell_i\) crossings.
	Formally, consider the dual \dgraph of the planarization \pgraph of \drawing, and for each vertex \(v \in V(G)\) let \(U_v \subseteq \dgraph\) be the set of vertices that correspond to cells \(\mathfrak{c}\) of \drawing such that $v$ lies on the boundary of $\mathfrak{c}$.
	We say a cell \(\mathfrak{c}\) of \drawing is \emph{\(s_it_i\)-far} if it corresponds to a vertex \(v_{\mathfrak{c}} \in V(\dgraph)\) 
	at distance more than \(\ell_i\) from \(U_{s_i}\) or \(U_{t_i}\),
	and \(\mathfrak{c}\) is \emph{far} if it is \(s_it_i\)-far for all \(i \in [k]\).
	Observe that in any solution of \SDE for \(\left(G,\drawing,\newedges,(\ell_i)_{i \in [|\newedges|]}\right)\) no drawing of an edge in \newedges can intersect far cells of \drawing.
	We refer to maximal unions of far cells in \drawing which form subsets of \(\mathbb{R}^2\) whose interior is connected as \emph{holes}.
	The interiors of holes are a natural choice for information that is not immediately relevant for the insertion of drawings for \newedges, in the sense that no intersections with these drawings can occur in far cells. However, as mentioned in the Proof Overview, omitting the interior of holes destroys information that identifies different parts of the same edge which are themselves not inside holes but are disconnected by a hole.
	
	To transfer this information between different parts of one edge---parts which could be crossed by a hypothetical solution but which are disconnected by holes---we introduce \emph{stitches} into the respective holes.
	More formally, for a hole \(H\) in \drawing we call an edge \(e\) in \(E(G)\) \emph{\(H\)-torn} if 
	\(e\) is split into at least two curves by the removal of the interior of \(H\) from \drawing.
	We call maximal subcurves of an \(H\)-torn edge after removing \(H\) \emph{(edge) parts} of \(e\)
	and refer to the endpoints of these subcurves as \emph{endpoints} of the corresponding edge part.
	Stitches will correspond to paths between the endpoints of edge parts of $H$-torn edges.
	To construct these paths we introduce
	so-called \emph{threads} which are edges that we insert into a hole \(H\) to connect parts of \(H\)-torn edges and
	derive the stitches from them by considering their planarization.
	
	To ensure that the obtained combinatorialization of $\drawing$ has bounded treewidth, the main goal of this section will be to 
	bound the number of stitches for each edge $s_it_i \in \newedges$ and hole $H$ by some function of $k + \ell_i$.
	We do this by considering which and how many edge parts of $H$-torn edges any simple $s_it_i$-curve in a hypothetical solution can cross.
	Here we face an apparent difficulty:
	it \emph{is} possible that there is an unbounded number of edge parts which
	are crossed by drawings of an added edge \(s_it_i\) in hypothetical solutions and
	each edge part
	belongs to a different \(H\)-torn edge (see Fig.~\ref{fig:unnecessary}).
	However, such situations can be safely avoided by restricting our attention to `reasonable' solutions, as we will see in Subsection~\ref{sec:detours}.
	In particular, to specify `reasonable' solutions we turn our attention to the behavior of drawings of added edges 
	in a hypothetical solution when they \emph{revisit} a cell of \drawing.
		Then, showing a bound on the number of stitches we introduce for an added edge \(s_it_i\) and hole \(H\) 
	is equivalent to showing that we can identify all but a bounded number of edges in \(E(G)\) 
	which are \(H\)-torn and cannot be crossed by a drawing of \(s_it_i\) in a `reasonable' hypothetical solution.
	This is what we focus on in Subsection~\ref{sec:stitches}.}%

	\later{%
	Let \(\left(G,\drawing,\newedges,(\ell_i)_{i \in [|\newedges|]}\right)\) be an instance of \SDE.
	We begin by identifying parts of \drawing that may be considered `unimportant' because they can never be intersected by the drawing of any of the edges \(s_it_i \in \newedges\) with at most \(\ell_i\) crossings.
	This is the case for every cell of \drawing which is separated by at least \(\ell_i\) cell boundaries from every cell containing \(s_i\) or every cell containing \(t_i\).
	More formally, consider the dual \dgraph of the planarization \pgraph of \drawing, and for each vertex \(v \in V(G)\) let \(U_v \subseteq \dgraph\) be the set of vertices that correspond to cells \(\mathfrak{c}\) of \drawing such that $v$ lies on the boundary of $\mathfrak{c}$.
	We say a cell \(\mathfrak{c}\) of \drawing is \emph{\(s_it_i\)-far} if it corresponds to a vertex \(v_{\mathfrak{c}} \in V(\dgraph)\) 
	at distance more than \(\ell_i\) from \(U_{s_i}\) or \(U_{t_i}\),
	and \(\mathfrak{c}\) is \emph{far} if it is \(s_it_i\)-far for all \(i \in [k]\).
	Observe that in any solution of \SDE for \(\left(G,\drawing,\newedges,(\ell_i)_{i \in [|\newedges|]}\right)\) no drawing of an edge in \newedges can intersect far cells of \drawing.
	We refer to maximal unions of far cells in \drawing which form subsets of \(\mathbb{R}^2\) whose interior is connected as \emph{holes}. 
	
	Similarly as in \cite{EibenGHKN20} we want to derive a planar graph of diameter bounded in \(k\) and \(\ell\) which serves as combinatorialization of \drawing on which we can invoke Courcelle's theorem; the diameter bound is what allows us to argue that the graph has bounded treewidth, a prerequisite for using Courcelle's theorem.
	To achieve bounded diameter it is necessary to omit certain parts of the original drawing from the combinatorialization.
	The interiors of holes are a natural choice for information that is not immediately relevant for the insertion of drawings for \newedges, in the sense that no intersections with these drawings can occur in far cells.
	While it is true that `omitting' the interior of holes from the combinatorialization of \drawing would allow us to correctly restrict the number of crossings of each added edge, unfortunately it means that we lose information that identifies different parts of the same edge which are themselves not inside holes but are disconnected by a hole; in particular, after omitting the holes we are no longer able to prevent double-crossings.
	
	
	To transfer this information between different parts of one edge---parts which could be crossed by a hypothetical solution could cross but which are disconnected by holes---we introduce \emph{stitches} into the respective holes.
	More formally, for a hole \(H\) in \drawing we call an edge \(e\) in \(E(G)\) \emph{\(H\)-torn} if 
	\(e\) is split into at least two curves by the removal of the interior of \(H\) from \drawing.
	We call maximal subcurves of an \(H\)-torn edge after removing \(H\) \emph{(edge) parts} of \(e\)
	and refer to the endpoints of these subcurves as \emph{endpoints} of the corresponding edge part.
	Note, that an endpoint of an edge part of an $H$-torn edge $e$ 
	is either a vertex incident to $e$ or a crossing point between $e$ and another edge in \drawing.
	Furthermore, each part of an \(H\)-torn edge has an endpoint on the boundary of \(H\).
	The aforementioned stitches will correspond to paths between the endpoints of edge parts of $H$-torn edges.
	To construct these paths we introduce
	so-called \emph{threads} which are edges that we insert into a hole \(H\) to connect parts of \(H\)-torn edges and
	derive the stitches from them by considering their planarization.
	In the end, stitches will ensure that we can relate two such different edge parts of the same \(H\)-torn edge in an MSO-encoding, which in turn allows us to prevent multiple crossings of the same edge in a solution.
	
	Of course, just connecting all the edge parts of all $H$-torn edges in this manner would not lead to a
	bounded treewidth graph.
	Hence, we will attempt to only introduce threads for parts of $H$-torn edges which may be crossed by  hypothetical solutions, with the aim of ensuring that the diameter (and hence the treewidth) of the combinatorization remains bounded by a function of our parameters.
	
	
	To achieve this, the main goal of this section will be to 
	bound the number of stitches for each edge $s_it_i \in \newedges$ and hole $H$ by some function of $k + \ell_i$.
	We do this by considering which and how many edge parts of $H$-torn edges any simple $s_it_i$-curve in a hypothetical solution can cross.
	Here we face an apparent difficulty:
	it \emph{is} possible that there is an unbounded number of edge parts which
	are crossed by drawings of an added edge \(s_it_i\) in hypothetical solutions and
	each edge part
	belongs to a different \(H\)-torn edge (see Fig.~\ref{fig:unnecessary}).
	However, such situations can be safely avoided by restricting our attention to `reasonable' solutions, as we will see in Subsection~\ref{sec:detours}.
	In particular, to specify `reasonable' solutions we turn our attention to the behavior of drawings of added edges 
	in a hypothetical solution when they \emph{revisit} a cell of \drawing.
		Then, showing a bound on the number of stitches we introduce for an added edge \(s_it_i\) and hole \(H\) 
	is equivalent to showing that we can identify all but a bounded number of edges in \(E(G)\) 
	which are \(H\)-torn and cannot be crossed by a drawing of \(s_it_i\) in a `reasonable' hypothetical solution.
	This is what we focus on in Subsection~\ref{sec:stitches}.}%
	
	\both{%
	\begin{figure}
		\begin{minipage}[t]{.47\textwidth}
			\centering
			\includegraphics{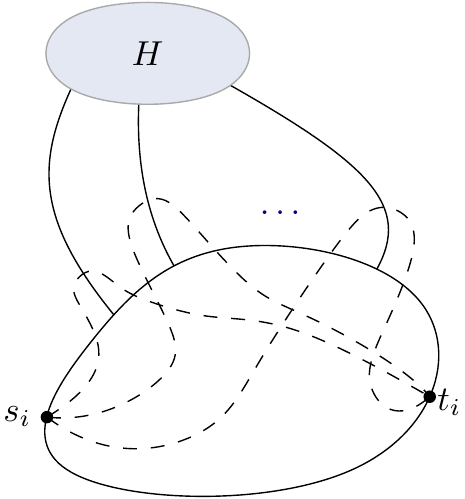}
			\subcaption{
				Assuming \(\ell_i = 3\),
				then the potential drawings of \(s_it_i\), depicted as dashed curves,
				can cross an arbitrary number of \(H\)-torn edges.
				\label{fig:unnecessary}}
		\end{minipage}
		\hfill
		\begin{minipage}[t]{.47\textwidth}
			\centering
			\includegraphics[page=2]{figures/unnecessary.pdf}
			\subcaption{
				Assuming \(\ell_i = \ell_j = 3\),
				then the drawing of $s_jt_j$ has to go through $\mathfrak c$ and
				the drawing of \(s_it_i\) has to revisit cell \(\mathfrak{c}\).
				\label{fig:necessary}}
		\end{minipage}
		\caption{Removable and unremovable detours. 
		}
		\label{fig:undetours}
	\end{figure}}


	\both{After adding stitches, we are finally able to define an appropriate combinatorialization of \drawing in Section~\ref{sec:patchworkdef} which we can use for the final application of Courcelle's theorem in Section~\ref{sec:MSO}.}
	
	\both{%
	\subsection{Detours and Reasonable Solutions}
	\label{sec:detours}}
	\short{%
Fix an added edge \(s_it_i\), a hole \(H\), and a cell \(\mathfrak{c}\) of the original drawing of \(G\). Note that a drawing of \(s_it_i\) in a hypothetical solution might revisit the cell \(\mathfrak{c}\) to avoid crossing the drawing of a different added edge \(s_jt_j\).
	Fig.~\ref{fig:necessary} exemplifies such a situation. Understanding how and why a solution might need to revisit a cell is a major component of our argument used to bound the number of stitches per hole. In fact, as we will see in this section, avoiding such crossings is the only reason why a cell might have to be revisited.}%
	\later{%
	Fix an added edge \(s_it_i\), a hole \(H\), and a cell \(\mathfrak{c}\) of the original drawing of \(G\).
	A drawing of \(s_it_i\) in a hypothetical solution might revisit the cell \(\mathfrak{c}\).
	At first glance this may seem counter-intuitive since then the drawing of \(s_it_i\) could be shortcut within \(\mathfrak{c}\) without increasing the number of crossings of \(s_it_i\).
	However, on second thought it becomes evident that the drawing of \(s_it_i\) might have to revisit \(\mathfrak{c}\) to avoid crossing the drawing of a different added edge \(s_jt_j\).
	Fig.~\ref{fig:necessary} exemplifies such a situation.
	Understanding how and why a solution might need to revisit a cell is a major component of our argument used to bound the number of stitches per hole. 
	In fact, as we will see in this section, avoiding such crossings is the only reason why a cell might have to be revisited.}%
	
	\both{%
	Let $\gamma$ be a drawing of \(s_it_i\) in a hypothetical solution which revisits \(\mathfrak{c}\). 
	A \emph{\(\mathfrak{c}\)-detour} (of \(\gamma\)) is a
	maximal subcurve of \(\gamma\) whose interior is disjoint from $\interior(\mathfrak{c})$ and has neither 
	\(s_i\) nor \(t_i\) as an endpoint. 
	Note that a \(\mathfrak{c}\)-detour might also consist of a singular point. 
	This case occurs when \(\gamma\) crosses an edge segment on the boundary of \(\mathfrak{c}\) which does not lie on the boundary of another cell.
	See Fig.~\ref{fig:detours} for an illustration.}%
	

\both{\begin{figure}
	\centering
	\begin{minipage}[t]{.49\textwidth}
		\includegraphics{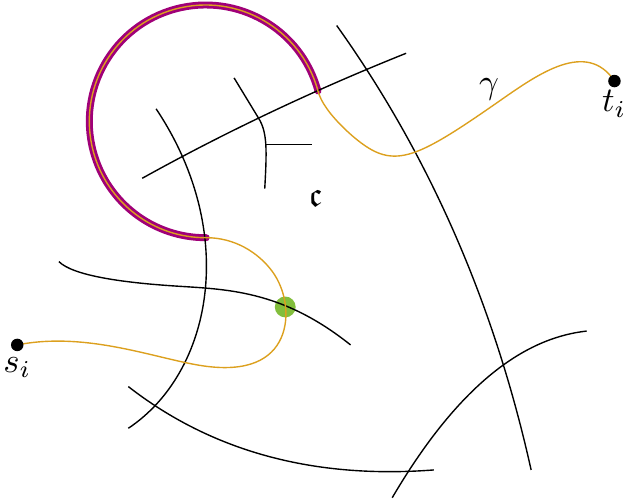}
		\subcaption{%
			Drawing \(\gamma\) of \(s_it_i\) in a hypothetical solution with two \(\mathfrak{c}\)-detours:
			\appendixandlongorshort{one is a curve (highlighted in purple) and the other is a  point (highlighted in green).}%
			{a curve (purple) and a point (green).}%
		}%
		\label{fig:detours}	
	\end{minipage}
	\hfill
	\begin{minipage}[t]{.45\textwidth}	
		\centering		
		\includegraphics[page=2]{figures/detours.pdf}
		\subcaption{
			\appendixandlongorshort{For the single point detour (green), 
				the avoided part of the boundary of \(\mathfrak{c}\) and the plane coincide and are dashed green.
				For the curve detour (purple), 
				the avoided part of the boundary of \(\mathfrak{c}\) is dashed purple 
				and the avoided region is shaded purple.}%
			{Two detours (green and purple) and their avoided parts of the boundary of \(\mathfrak{c}\) (dashed).}%
			\label{fig:avoided}}
	\end{minipage}
   \caption{Detours and the parts defined by them.}
\end{figure}}

%

	\both{\begin{definition}\label{def:detour}
		Let $\delta$ be a \(\mathfrak{c}\)-detour,
		and let the embedding $\mathcal{E}$ consist only of \(\delta\) and the restriction of \drawing to the boundary of \(\mathfrak{c}\).
		Then \(\delta\) partitions the boundary of \(\mathfrak{c}\) into two connected parts: 
			the part incident to the unbounded (i.e.\ outer) cell in \(\mathcal{E}\), and 
			the \emph{\(\delta\)-avoided part} which is not incident to the outer cell in \(\mathcal{E}\).
		
		Additionally, we call the subset of \(\mathbb{R}^2\) which is enclosed by \(\delta\) and the \(\delta\)-avoided part of the boundary of \(\mathfrak{c}\) together with the \(\delta\)-avoided part of the boundary of \(\mathfrak{c}\) itself the  \emph{\(\delta\)-avoided region}.
		See Fig.~\ref{fig:avoided} for an illustration.
		A \(\mathfrak{c}\)-detour \(\delta\) is \emph{unremovable} 
		if there exists an added edge \(s_jt_j\) with \(j \neq i\) such that exactly one of \(s_j\) and \(t_j\) 
		lies in the \(\delta\)-avoided region of \drawing. 
		In that case we say that the endpoint (\(s_j\) or \(t_j\)) in the \(\delta\)-avoided region is \emph{avoided} by $\delta$, or that \(\delta\) is \emph{around} the endpoint. 
		We call a \(\mathfrak{c}\)-detour \emph{removable} if it is not unremovable.
		\end{definition}}

\later{
		Note that we can generalize the definitions of \(\zeta\)-avoided parts of cell boundaries and regions 
		for an arbitrary simple curve \(\zeta\) that starts and ends on the boundary of the same cell of \drawing and whose interior does not intersect that cell.

	The above definition of necessity of detours is justified by the following result.
	Intuitively it states that whenever a solution uses a \(\mathfrak{c}\)-detour for the drawing of an added edge this has to be because it is necessary in the sense that it avoids connecting immediately though \(\mathfrak{c}\) because of another added edge.
	Otherwise it can be short-cut through \(\mathfrak{c}\) (possibly after short-cutting some other removable \(\mathfrak{c}\)-detours).}
	
\both{%
	\begin{lemma}[$\star$]
		\label{lem:unnecessary}
		If there is a solution, then there exists a solution in which no drawing of any added edge contains a removable \(\mathfrak{c}'\)-detour for any cell \(\mathfrak{c}'\) of \drawing.
	\end{lemma}}
\later{%
	\begin{proof}
		Consider a solution which minimizes the sum of the number of detours over all drawings of added edges and cells \(\mathfrak{c}'\).
		Assume \(\mathfrak{c}'\) is a specific cell at which there is a removable \(\mathfrak{c}'\)-detour \(\delta\) as part of the drawing of added edge \(s_jt_j\).
		Since \(\mathfrak{c}'\) is a cell in \drawing, there is a curve in the interior of \(\mathfrak{c}'\) that connects the first crossing point \(d_1\) of the drawing of \(s_jt_j\) into \(\mathfrak{c}'\) and the last crossing point \(d_2\) of the drawing of \(s_jt_j\) out of \(\mathfrak{c}'\) and does not intersect any edge in \(E(G)\).
		If no drawing of another added edge in the considered solution separates \(d_1\) from \(d_2\) within \(\mathfrak{c}'\), this simple curve could be used instead of \(\delta\) contradicting the minimality assumption on the considered solution.
		Hence, let \(s_{j'}t_{j'} \neq s_jt_j\) be another added edge whose drawing in the considered solution separates \(d_1\) from \(d_2\) within \(\mathfrak{c}'\).
		Since we are assuming \(\delta\) to be removable the drawing of \(s_{j'}t_{j'}\) also has to cross \(\delta\), or it also contains a \(\mathfrak{c}'\)-detour, on which we can iterate the argument.
		After at most \(k \cdot \ell\) iterations, in each of which we always find some detour which does not intersect the one we started the respective iteration with, we arrive at the case that the first entry and last exit point \(d_1\) and \(d_2\) of the considered \(\mathfrak{c}'\)-detour \(\delta\) are only separated in \(\mathfrak{c}'\) by drawings of added edges that also cross \(\delta\).
		This simple bound\footnote{A better bound can also be argued, but since we do not use this result algorithmically we do not give this argument here, as it would not improve our results.} on the number of iterations can easily be argued because each edge can have at most \(\ell\) crossings, and hence also at most this many different detours are part of the drawing of an added edge in the considered solution.
		
		Now it is easy to see that there is a \(d_1\)-\(d_2\)-curve within \(\mathfrak{c}'\) that crosses no drawing of an edge in \(E(G)\) and at most the same drawings of added edges in the considered solution as \(\delta\) does at most once.
		This means we can replace \(\delta\) by such a curve in our initial solution to obtain a solution with a smaller total number of removable detours, contradicting the minimality assumption.
	\end{proof}}
	
\short{
	
Lemma~\ref{lem:unnecessary} allows us to restrict our attention to solutions which do not contain any removable detours (these are the solutions we intuitively referred to as `reasonable').}

\later{
	
	Justified by this previous lemma, we want to from now understand a `reasonable' solution to be a solution 	in which no drawing of any added edge contains a removable \(\mathfrak{c}'\)-detour for any cell \(\mathfrak{c}'\) of \drawing, and 
	can safely restrict ourselves to computing such reasonable solutions.}

	\both{\subsection{Defining and Finding Stitches}
	\label{sec:stitches}}
	\short{%
	Let $s_it_i \in \newedges$ and $\hole$ be a hole. Our goal now will be to compute
	the set of edge parts in \(E(G)\) (w.r.t. $\hole$) which could be
	crossed by a drawing of \(s_it_i\) in some reasonable hypothetical solution, and to show that the number of such edge parts is bounded by our parameters.
	As we obviously do not know any hypothetical solution we 
		cannot compute this set directly.
	Consequently, we identify and compute a slightly larger set: the set of all edge parts that can be crossed by some so-called \emph{solution curve} for $s_i$ and $t_i$ that is 
	superficially like an \(s_it_i\)-curve in a `reasonable' hypothetical solution (but which might induce double-crossings).}%
	\later{%
	Throughout this section we fix an added edge $s_it_i \in \newedges$ and a hole $\hole$.
	Our goal is to compute
	a bounded number of edge parts in \(E(G)\) which could be
	crossed by a drawing of \(s_it_i\) in some reasonable hypothetical solution.
	As we obviously do not know any hypothetical solution we 
	cannot directly define an efficiently checkable notion of crossability by
	demanding that an $H$-torn edge is crossed by the drawing of \(s_it_i\) in some reasonable hypothetical solution.
	Consequently, we only demand that a crossable edge part of an $H$-torn edge can be
	crossed by some so-called \emph{solution curve} for $s_it_i$ that 
	behaves superficially like a drawing of \(s_it_i\) in a reasonable hypothetical solution.
	Specifically, we will not restrict solution curves to not contain multiple crossings with the same edge in \drawing.}

	\both{%
		\begin{definition}
		\label{def:solcurve}
		\appendixandlongorshort{
			A \emph{solution curve} for \(s_it_i\) is a simple curve \(\gamma\) that
			\begin{itemize}
				\item starts in \(s_i\) and ends in \(t_i\);
				\item produces at most $\ell_i$ crossings with \drawing; and
				\item whenever \(\gamma\) intersects a cell \(\mathfrak{c}'\) in more than one maximal connected subcurve there is an added edge \(s_jt_j\) with \(j \neq i\) such that exactly one of \(s_j\) and \(t_j\) lies in the \(\zeta\)-avoided part of \drawing, where \(\zeta\) is a maximal connected subcurve of \(\gamma\) outside of \(\mathfrak{c}'\) between two intersections of \(\gamma\) with \(\mathfrak{c}'\).
			\end{itemize}
			A part of an \(H\)-torn edge \(e \in E(G)\) is \emph{crossable} for \(s_it_i\) if it is crossed by a solution curve for \(s_it_i\).
		}{%
		A \emph{solution curve} for \(s_it_i\) is a simple curve \(\gamma\) that $(i)$ starts in \(s_i\) and ends in \(t_i\); $(ii)$ produces at most $\ell_i$ crossings with \drawing; and
$(iii)$ whenever \(\gamma\) intersects a cell \(\mathfrak{c}'\) in more than one maximal connected subcurve there is an added edge \(s_jt_j\) with \(j \neq i\) such that exactly one of \(s_j\) and \(t_j\) lies in the \(\zeta\)-avoided part of \drawing, where \(\zeta\) is a maximal connected subcurve of \(\gamma\) outside of \(\mathfrak{c}'\) between two intersections of \(\gamma\) with \(\mathfrak{c}'\).
		A part of an \(H\)-torn edge \(e \in E(G)\) is \emph{crossable} for \(s_it_i\) if it is crossed by a solution curve for \(s_it_i\).
		}%
	\end{definition}}

	
	\both{%
	\begin{lemma}[$\star$]	
		\label{lem:crossable}
		For every hole $\hole$ and every added edge $s_it_i\in F$
		there are less than \[\ell_i(2\ell_i + 1)! \cdot \left(4 k (\ell_i+2) (\ell_i + 1)^{\ell_i + 1} \right)^{2\ell_i + 1}\] parts of \(H\)-torn edges that are crossable for \(s_it_i\).
	\end{lemma}}%
	
	\short{%
	\begin{proof}[Proof Sketch]
		We show that there is a set $K$ of less than $(2\ell_i + 1)!\left(4 k (\ell_i+2) (\ell_i + 1)^{\ell_i + 1} \right)^{2\ell_i + 1}$ 
		solution curves for $s_it_i$ such that each crossable edge part for $s_it_i$ is crossed by at least one of the curves in $K$.
		Then the claim follows as each solution curve crosses at most \(\ell_i\) edges.
		
		Assume for contradiction that the minimum set $K$ that witnesses crossability of parts of \(H\)-torn crossable edges for \(s_it_i\) consists of at least \((2\ell_i + 1)!\left(4 k (\ell_i+2) (\ell_i + 1)^{\ell_i + 1} \right)^{2\ell_i + 1}\) solution curves for \(s_it_i\).		
		Consider the restricted drawing \(\drawing_{H}\) which is given by \drawing restricted to the boundary of \(H\), all \(H\)-torn edges in \(E(G)\), as well as $s_i$ and~$t_i$.

		We associate each $s_it_i$ curve in $K$ with the set of cells of \(\drawing_{H}\) which it intersects and the set of edge segments in \(\planarization{\drawing_{H}}\) which it crosses.
			In this way, each curve in \(K\) is associated to a set of size at most \(2\ell_i + 1\).
			By the minimality of $K$, no two curves in \(K\) are associated to the same set 
			of cells and edge segments.	
		Using the sunflower lemma~\cite{Erdos60,FlumGrohe06} for the set system 
		given by the sets of cells associated to the $s_it_i$ curves in \(K\) 
		we obtain a set of at least \(4 k (\ell_i+2) (\ell_i + 1)^{\ell_i + 1}\) solution curves \(K^{\sun} \subseteq K\)
		which all intersect pairwise different cells of \(\drawing_{H}\) and edge segments of~\(\planarization{\drawing_{H}}\),
		apart from the cells and edge segments in the center of a sunflower, which they all intersect.
		Moreover, as curves in $K$ intersect at most $\ell_i$ edges
		we find at most $\ell_i + 1$ cells in the center.
		
		By the pigeonhole principle 
		there is a set of at least \(4k(\ell_i + 2)\) curves in \(K^{\sun}\) which all intersect the cells in the center of the sunflower in the same order (taking into account repetitions of cells). 
		Let $K^{\sun}_\sigma \subseteq K^{\sun}$ be such a set of curves and let \(\sigma = \mathfrak{c}_1, \dotsc, \mathfrak{c}_l\) with \(l \leq \ell_i + 1\) be the order in which these curves traverse the cells in the center of the sunflower.
		
		As each \(\mathfrak{c}_j\) with \(j \in [l]\) is a cell in a restriction of \drawing containing all \(H\)-torn edges, no part of an \(H\)-torn edge intersects the interior of \(\mathfrak{c}_j\).
		In particular parts of \(H\)-torn edges are not crossed by any curve in \(K^{\sun}_\sigma\) within \(\interior(\mathfrak{c_j})\).
		
		When considering subcurves of curves \emph{between} each \(\mathfrak{c}_j\) and \(\mathfrak{c}_{j + 1}\), we can find at most \(4(k - 1)\) `extremal' such subcurves which separate all other subcurves from \(H\) together with \(\mathfrak{c}_j\) and \(\mathfrak{c}_{j + 1}\).
		These extremal subcurves together cross any crossable edge part of an \(H\)-torn edge intersected by any other considered subcurve.
		See Fig.~\ref{fig:bound_per_hole} for an illustration.
	
		In this way we find at least one curve after the removal of which from \(K\) the same crossable edge parts of \(H\)-torn edges are intersected, contradicting our minimality assumption.
	\end{proof}}
	
	\both{\begin{figure}
		\includegraphics[width=\textwidth]{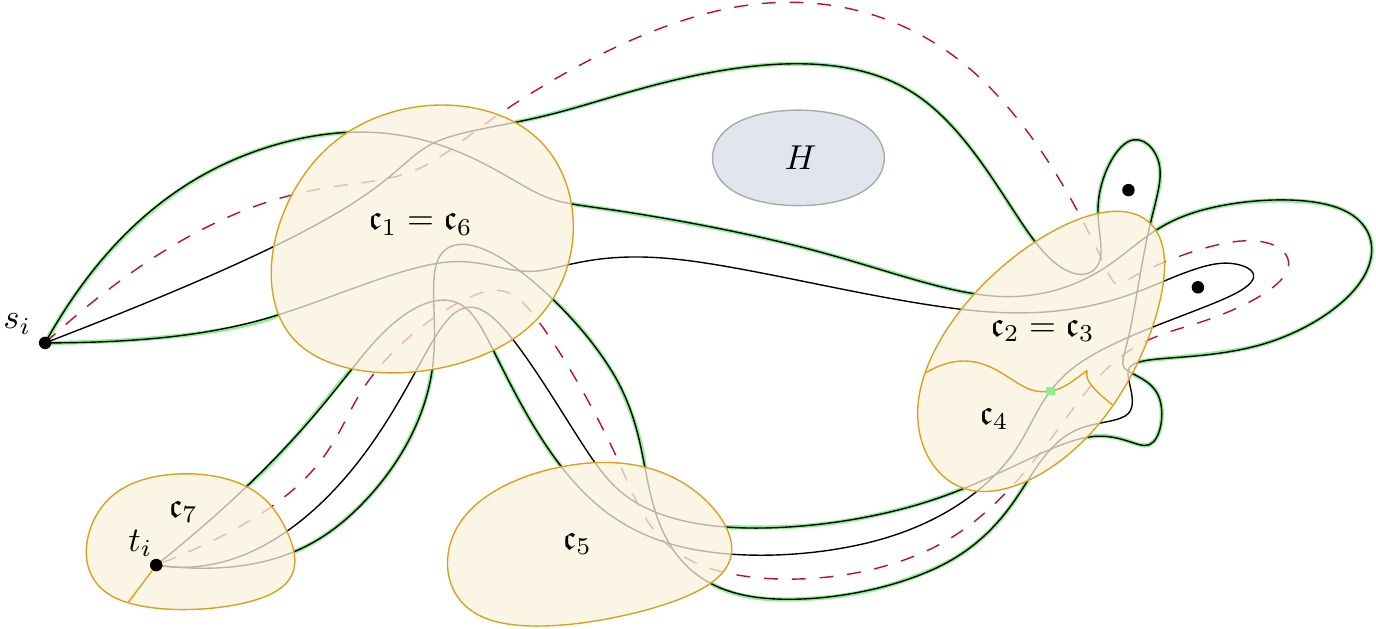}
		\caption{The cells $\mathfrak{c}_1,\ldots,\mathfrak{c}_7$ are in the sunflower center. The red dashed $s_it_i$ curve cannot be part of the minimal set of curves $K$. The extremal subcurves are highlighted in green. \label{fig:bound_per_hole}}
	\end{figure}}
	
	\later{%
	\begin{proof}
		We will show that there is a set $K$ of less than $(2\ell_i + 1)!\left(4 k (\ell_i+2) (\ell_i + 1)^{\ell_i + 1} \right)^{2\ell_i + 1}$ solution curves for $s_it_i$ that together witness the crossability for \(s_it_i\) of each part of an \(H\)-torn edge that is crossable for \(s_it_i\),
		i.e.\ such that each edge part that is crossable for $s_it_i$ is crossed by at least one of the curves in $K$.
		Then the claim immediately follows since each solution curve crosses at most \(\ell_i\) edges.
		
		Assume for contradiction that the minimum set $K$ that witnesses crossability of parts of \(H\)-torn crossable edges for \(s_it_i\) consists of at least \((2\ell_i + 1)!\left(4 k (\ell_i+2) (\ell_i + 1)^{\ell_i + 1} \right)^{2\ell_i + 1}\) solution curves for \(s_it_i\).		
		Consider the restricted drawing \(\drawing_{H}\) which is given by \drawing restricted to the boundary of \(H\), all \(H\)-torn edges in \(E(G)\), as well as $s_i$ and~$t_i$.
		
		Each $s_it_i$ curve in \(K\) can be associated to the set of cells of \(\drawing_{H}\) which it intersects and the set of edge segments in the planarization of \(\drawing_{H}\) which it crosses.
		In this way, each $s_it_i$ curve in \(K\) is associated to a set of size at most \(2\ell_i + 1\), 
		since it can have at most $\ell_i$ crossings with \(\drawing_{H}\). 
		By the minimality of $K$, no two curves in \(K\) can be associated to the same set of edge segments, which also means that each curve in \(K\) is associated to a different set of cells and edge segments.
		
		Now we can apply the sunflower lemma to the set system given by the sets of cells associated to the $s_it_i$ curves in \(K\) to obtain a sunflower of size at least \(4 k (\ell_i+2) (\ell_i + 1)^{\ell_i + 1}\).
		This sunflower corresponds to a set of at least \(4 k (\ell_i+2) (\ell_i + 1)^{\ell_i + 1}\) solution curves \(K^{\sun} \subseteq K\)  
		which all intersect pairwise different cells of \(\drawing_{H}\) and edge segments of the planarization of~\(\drawing_{H}\), apart from the cells and edge segments in the center of the sunflower, which they all intersect.
		Since the curves in $K$ intersect at most $\ell_i$ edges, 
		the number of cells in the center of the sunflower is upper bounded by \(\ell_i + 1\). 
		By the pigeonhole principle 
		there is a set of at least \(4k(\ell_i + 2)\) curves in \(K^{\sun}\) which all intersect the cells in the center of the sunflower in the same order (taking into account that the curves might intersect the same cell multiple times). 
		Let $K^{\sun}_\sigma \subseteq K^{\sun}$ be such a set of curves and let \(\sigma = \mathfrak{c}_1, \dotsc, \mathfrak{c}_l\) with \(l \leq \ell_i + 1\) be the order (with repetitions) in which these curves traverse the cells in the center of the sunflower.
		To unify notation, when the center of the sunflower does not contain any cell we define \(\mathfrak{c}_1 = \mathfrak{c}_2 = \mathfrak{c}\).  
		
		As each \(\mathfrak{c}_j\) with \(j \in [l]\) is a cell in a restriction of \drawing containing all \(H\)-torn edges, no part of an \(H\)-torn edge intersects the interior of \(\mathfrak{c}_j\).
		In particular parts of \(H\)-torn edges are not crossed by any curve in \(K^{\sun}_\sigma\) within \(\interior(\mathfrak{c_j})\).
		In that sense, we are only interested in the subcurves of the curves in \(K^{\sun}_\sigma\) within \(\interior(\mathfrak{c_j})\) outside of \(\bigcup_{1\le j \le l} \interior(\mathfrak{c}_j)\). 
		The connected parts of the embeddings in $\mathbb{R}^2 \setminus \bigcup_{1\le j \le l} \interior(\mathfrak{c}_j)$ of the curves in $K^{\sun}_\sigma$ are pairwise non-intersecting, as intersections are only possible inside cells in the center of the sunflower.
		See Fig.~\ref{fig:bound_per_hole} for an illustration.		
		
		Consider the subcurves of curves in $K^{\sun}_\sigma$ from $s_i$ to the first point they share with the boundary of $\mathfrak{c}_1$, if \(\mathfrak{c}_1 \neq \mathfrak{c}\). 
		If $s_i$ is inside or on the boundary of $\mathfrak{c}_1$ the interior of these subcurves does not intersect any $\hole$-torn edge. 
		Otherwise, if $s_i$ does not intersect $\mathfrak{c}_1$, 
		those subcurves of curves in $K^{\sun}_\sigma$ 
		are pairwise non-intersecting. 
		As in the case in which the center of the sunflower was empty, 
		we can identify at most two \emph{extremal} subcurves of two curves $\zeta_1$ or $\eta_1$
		that together with $\interior (\mathfrak{c}_1)$ separate $H$ from the rest of the subcurves without $s_i$. 
		Since any $\hole$-torn crossable edge has to reach $\hole$ and does not subdivide $\mathfrak{c}_1$, 
		these two extremal subcurves must cross any crossable edge intersected by any other subcurve. 
		Analogously,
		we can identify at most two extremal subcurves of curves in $K^{\sun}_\sigma$ between $t_i$ and the boundary of $\mathfrak{c}_l$ that belong to two curves $\zeta_{l+1}$ or $\eta_{l+1}$.  
		%
		
		For the remaining subcurves of curves in \(K^{\sun}_\sigma\) within \(\interior(\mathfrak{c_j})\) outside of \(\mathbb{R}^2 \setminus \bigcup_{1\le j \le l} \interior(\mathfrak{c}_j)\) we introduce the following formulation.
		For each curve $\gamma$ in~$K^{\sun}_\sigma$, the subcurve of $\gamma$ \emph{between} the boundaries of $\mathfrak{c}_{j}$ and $\mathfrak{c}_{j+1}$ with $1\le j \le l-1$ refers to the maximal connected part of $\gamma$ whose interior is disjoint from $\bigcup_{1\le j \le l} \interior(\mathfrak{c}_j)$ (this might be a single point) 
		between $\mathfrak{c}_{j}$ and $\mathfrak{c}_{j+1}$ 
		in the traversal order specified by $\sigma$. 
		
		For the case in which $\mathfrak{c}_{j} \neq \mathfrak{c}_{j+1}$,  
		as in the previous cases, we can identify at most two extremal subcurves of two curves $\zeta_{j+1}$ or $\eta_{j+1}$ between \(\mathfrak{c}_j\) and \(\mathfrak{c}_{j + 1}\)
		that together with $\interior (\mathfrak{c}_{j})$ and $\interior (\mathfrak{c}_{j+1})$ separate $H$ from all other such subcurves. 
		Again, since any crossable edge part has to connect to $\hole$ and neither subdivides $\mathfrak{c}_{j}$ nor $\mathfrak{c}_{j+1}$, 
		the two extremal subcurves must intersect any crossable edge part intersected by any other such subcurve. 	

		We have to take a little more care in the case in which $\mathfrak{c}_{j} = \mathfrak{c}_{j+1}$ (in particular this includes the case that \(\mathfrak{c}_1 = \mathfrak{c}_2 = \mathfrak{c}\)).
		In this case, the subcurves of curves in $K^{\sun}_\sigma$ between $\mathfrak{c}_{j}$ and $\mathfrak{c}_{j+1}$ are a maximal connected subcurves outside of \(\mathfrak{c}_j\) between two intersections with \(\mathfrak{c}_j\) (note that these maximal connected subcurves can be singular points).
		Then by the definition of solution curves for each considered subcurve \(\gamma'\) of a solution curve \(\gamma \in K^{\sun}_\sigma\), there is an added edge \(s_{j'}t_{j'}\) with \({j'} \neq i\) such that exactly one of \(s_{j'}\) and \(t_{j'}\) lies in the \(\gamma'\)-avoided part of \drawing.
		Subcurves of curves in \(K^{\sun}_\sigma\) between \(\mathfrak{c}_j\) and \(\mathfrak{c}_{j + 1}\) which avoid a specific endpoint of one of the added edges \(s_{j'}t_{j'}\) with \(j' \neq i\)
		are pairwise non-intersecting and they are nested in the sense that,  
		given two detours, their avoided regions must be sub- or super sets of each other. 
		For each of the \(2k - 1\) such endpoints $p$, we can identify at most two extremal subcurves $\delta_{j+1}^p$ and $\partial_{j+1}^p$ around $p$, which are subcurves of solution curves \(\zeta_{j+1}^p\) and \(\eta_{j+1}^p\) respectively.
		More precisely, 
		$\delta_{j+1}^p$ does not contain $H$ in the $\delta_{j+1}^p$-avoided part of the plane and 
		is inclusion-maximal with this property  
		(it contains in the $\delta_{j+1}^p$-avoided region all other considered subcurves around $p$ that do not contain $H$ in their avoided regions). 
		Similarly, $\partial_{j+1}^p$ contains $H$ in the $\partial_{j+1}^p$-avoided region and is inclusion-minimal with this property.
		Any point on a non-extremal subcurve around~$p$ 
		is separated from \(H\) by  $\delta_{j+1}^p$,  $\partial_{j+1}^p$, and $\interior (\mathfrak{c}_{j})$. 
		Thus, as above, the extremal subcurves around $p$ together must cross any crossable edge part intersected by any other considered subcurve around~$p$.

%
		

		Let $\zeta_j$, $\eta_j$, $\zeta^{p}_j$, and $\eta^{p}_j$ be defined as above (if they exist and ignored otherwise).
		for \(j \in [l+1]\) and \(p \in \{s_j,t_j \mid j \neq i\}\). 
		Then it holds that any \(\gamma \in {K}^{\sun}_{\sigma} \setminus (\{\zeta_j,\eta_j,\zeta^{p}_j,\eta^{p}_j \mid j \in [l+1], \text{$p$ endpoint of edge in $F\setminus s_it_i$}\})\) 
		only crosses crossable edge parts which are also crossed by at least one of $\{\zeta_j,\eta_j,\zeta^{p}_j,\eta^{p}_j \mid j \in [l+1], p \in \{s_j,t_j \mid j \neq i\}\})$.
		Since \(|{K}^{\sun}_{\sigma} \setminus (\{\zeta_j,\eta_j,\zeta^{p}_j,\eta^{p}_j \mid j \in [l+1], p \in \{s_j,t_j \mid j \neq i\}\})| \geq 4k(\ell_i + 2) - (l+1)\max \{2, 4(k-1) \} \geq (\ell_i + 2) (4k - \max \{2, 4(k-1) \}) > 0\) we get a contradiction to the minimality assumption on \(K\).
	\end{proof}}

	
	\both{
		
		While the fact that the number of crossable edge parts we want to introduce stitches for is bounded by a function in our parameters is reassuring, we need to be able to actually introduce these stitches before being able to give our final MSO encoding of hypothetical solutions.
		For this we invoke Courcelle's theorem in Lemma~\ref{lem:findcrossable} independently of its final application.
		This then allows us to insert the corresponding stitches.}
		
		\both{%
		\begin{lemma}[$\star$]		
			\label{lem:findcrossable}
			There is a fixed-parameter algorithm parameterized by $k+\ell$ which identifies, for an added edge \(s_it_i\) and a hole \(H\), all parts of \(H\)-torn edges which are crossable for \(s_it_i\).			
		\end{lemma}}
		
		\short{%
		\begin{proof}[Proof Sketch]
			Consider the \emph{combined primal-dual graph} \(G^p_d\) of \pdrawing without the interior of holes
			that arises from \pdrawing by
			\begin{enumerate}
				\item removing the edge segments in the interior of each hole, rendering each hole a single face;
				\item subdividing each edge \(e\) of \(\pgraph\) by a new vertex labeled as \emph{edge vertex}; and
				\item inserting a vertex labeled as \emph{face vertex} for every face of the graph obtained till now, and connecting that vertex to all vertices on the boundary of that face.
			\end{enumerate}
			\(G^p_d\) is planar and has diameter
			at most
			\(4(\ell_i + 1)\),
			and hence has
			treewidth at most $3(4(\ell_i + 1))$~\cite{RobertsonS84}. To complete the proof,
			we
			construct an MSO formula that
			checks whether each part of an $H$-torn edge is crossable. In particular, the formula considers each edge vertex in \(G^p_d\) and encodes the existence of an $s_i$-$t_i$ walk of length $\bigoh(\ell_i)$ in \(G^p_d\) that satisfies an additional, technical condition capturing the restriction imposed on solution curves in Definition~\ref{def:solcurve}.
		\end{proof}}
		
		\later{%
		\begin{proof}
			Using binary search we can find the correct number of parts of \(H\)-torn edges which are crossable in \(\log(f(k,\ell_i)g(k,\ell_i,n))\) time, where \(f(k,\ell_i)\) is the bound we obtain from Lemma~\ref{lem:crossable}, assuming \(g(k,\ell_i,n)\) is the time required to find a fixed number of crossable edge parts or decide that fewer crossable edge parts exist.
			We will find a fixed number \(q\) of crossable edges, by encoding the property of being the endpoint of a crossable edge part on the boundary of \(H\) within a graph of bounded treewidth.
			We can then find a satisfying assignment for up to \(2q\) free variables which have that property using Courcelle's theorem, completing the proof of the lemma.
			
			Consider the \emph{combined primal-dual graph} \(G^p_d\) of \pdrawing without the interior of holes
			that arises from \pdrawing by
			\begin{enumerate}
				\item removing the edge segments in the interior of each hole, rendering each hole a single face;
				\item subdividing each edge \(e\) of \(\pgraph\) by a new vertex labeled as \emph{edge vertex}; and
				\item inserting a vertex labeled as \emph{face vertex} for every face of the graph obtained till now, and connecting that vertex to all vertices on the boundary of that face.
			\end{enumerate}
			Obviously \(G^p_d\) is planar.
			Moreover, one can argue that the diameter of \(G^p_d\) is upper-bounded by \(4\ell_i + 4\);
			this is because every vertex in \(G^p_d\) is at distance at most one from a face vertex, and every face vertex is at distance at most \(2\ell_i + 1\) from \(s_i\) because the interiors of holes were deleted.
			Together this implies that \(\operatorname{tw}(G^p_d) \leq 3(4\ell_i + 4)\)~\cite{RobertsonS84}.

			It remains to show MSO encodability of the property of being the endpoint of a crossable edge part on the boundary of \(H\) in the graph \(G^p_d\) in which we additionally introduce an auxiliary labeling to describe paths in \(G^p_d\)
			that follow edge parts.
			Since a labeling (the so called \emph{tracking labeling}) for following edge parts will also be used in the final application of Corcelle's theorem in Section~\ref{sec:patchworkdef} we refer to that section for details.
			
			One can easily encode the property of being a vertex \(v\) on the boundary of \(H\), simply as being a neighbor of the face vertex that corresponds to \(H\); and using the additional labeling the property of being connected to some other vertex via an edge part.
			Such an \emph{other vertex} \(x\) should be an edge vertex which witnesses crossability of the edge part in question, i.e.\ we want there to be a walk in \(G^p_d\) that corresponds to a solution curve for \(s_it_i\) and contains \(x\).
			This means it suffices to encode walks which correspond to solution curves for \(s_it_i\).
			
			Because of the restriction on the number of crossings with \drawing of a solution curve, a solution curve always corresponds to a walk of length at most \(\mathcal{O}(\ell_i)\) in \(G^p_d\), which means we can just use free variables for all its vertices and edges.
			It is trivial to ensure that a walk starts in \(s_i\) and ends in \(t_i\).
			
			We can also ensure that the inner vertices of the walk follow a pattern of using a face vertex, followed by an edge vertex, and so on, the last inner vertex of the walk being a face vertex.
			To make sure the walk corresponds to a \emph{simple} curve, we have to make sure that the walk does not force a crossing of a curve that conforms to this walk.
			This can be done by including the condition that there are no edge vertices \(e_1, e_2, e_3, e_4\) which
			\begin{enumerate}
				\item are visited in the order \(e_1, e_2, e_3, e_4\) by the walk and are all in the neighborhood of the same face vertex \(f\) (this means than they correspond to edge segments on the boundary of the same face); and
				\item occur on a cycle in \(G^p_d\) that consists only of neighbors of \(f\) in the order \(e_1, e_3, e_2, e_4\).
			\end{enumerate}
			
			The last condition we need to encode is the fact that a curve \(\gamma\) corresponding to our walk \(\rho_{\gamma}\) should satisfy that
			whenever \(\gamma\) intersects a cell \(\mathfrak{c}'\) in more than one maximal connected subcurve there is an added edge \(s_jt_j\) with \(j \neq i\) such that exactly one of \(s_j\) and \(t_j\) lies in each \(\zeta\)-avoided part of \drawing, where \(\zeta\) is a maximal connected subcurve of \(\gamma\) outside of \(\mathfrak{c}'\) between two intersections of \(\gamma\) with \(\mathfrak{c}'\).
			For all cells \(\mathfrak{c}'\) the fact that \(\gamma\) intersects \(\mathfrak{c}'\) in more than one maximal connected subcurve will be equivalent to the face vertex \(v_{\mathfrak{c}'}\) corresponding to \(\mathfrak{c}'\) occurring multiple times on the walk, and each \(\zeta\) will correspond to the curve given by a maximal subwalk \(\rho_{\zeta}\) between two consecutive occurrences of \(v_{\mathfrak{c}'}\) on the walk.
			Using this characterization, all \(\rho_{\zeta}\) can be quantified over in an MSO formula.
			
			For each \(\rho_{\zeta}\) we can encode the equivalent of the \(\zeta\)-avoided part of the boundary of \(\mathfrak{c}'\) in \(G^p_d\) by defining it as the set of vertices \(A_{\zeta}\) which are neighbors of \(v_{\mathfrak{c}'}\), such that there is no path that connects these vertices to the face vertex corresponding to the outer cell of \drawing without the interior of holes, without intersecting \(\rho_{\gamma}\) or the closed neighborhood of \(v_{\mathfrak{c}'}\).
			Then the equivalent of the \(\rho_{\zeta}\)-avoided region in the patchwork graph is the subgraph of \(G^p_d\) which is separated from the face vertex of the outer cell of \drawing without the interior of holes by \(A_{\zeta}\) and \(\rho_{\zeta}\).
			All of this is MSO encodable, which means that we can also require that there is an added edge \(s_jt_j\) with \(j \neq i\) such that exactly one of \(s_j\) and \(t_j\) lies in the equivalent of the \(\zeta\)-avoided part in \(G^p_d\).
		\end{proof}}
	
		\both{%
		\begin{definition}
			For a hole \(H\) in \drawing\ and an added edge
			\(s_it_i \in \newedges\), a \emph{thread} is a pair of two endpoints of two distinct	
			edge parts of the same \(H\)-torn edge in \(e \in E(G)\) satisfying the following properties: (i) both edge parts are crossable for \(s_it_i\), (ii) there is no other crossable edge part between these edge parts along a traversal of \(e\), and (iii) there is no other endpoint of one of the two edge parts along a traversal of \(e\).
			We denote the set of all threads for \(H\) and \(s_it_i\) as \(T_{H,s_it_i}\), and define the set of all threads for \(H\) as \(T_H = \bigcup_{i \in [k]} T_{H,s_it_i}\).
		\end{definition}}

		\later{%
		\begin{observation}
				\label{obs:independent}
				Let $H$ be a hole in \drawing with not connected boundary and 
				$T_{H}$ the set of threads we computed for $H$.
				For an edge $s_it_I$ and thread $uv \in T_{H,s_i,t_i}$ between vertices $u$ and $v$ on the boundary of $H$
				we find that $u$ and $v$ are always in a connected part of the boundary of $H$.
		\end{observation}}
	
	\both{
		
		An embedding of $T_H$ is a set of curves, contained completely in $H$, which connect each pair of two endpoints of edge parts in $T_H$.}
	
	\both{%
			\begin{lemma}[$\star$]
			\label{lem:stitches}
			There is a fixed-parameter algorithm parameterized by $k+\ell$ that computes, for a hole \(H\) in \drawing, a simple embedding of $T_H$.
		\end{lemma}}
	\later{%
		\begin{proof}
			We can iterate over \(i \in [k]\) and compute \(T_{H,s_it_i}\) by invoking Lemma~\ref{lem:findcrossable} to find all parts of \(H\)-torn edges which are crossable for \(s_it_i\) in \FPT\ time, and then traversing the \(H\)-torn edges to determine which of them are connected by threads.
			This immediately allows to compute \(T_H\).
			
			We also have to find an embedding of $T_H$ into each $H$.
			Since no thread in $T_H$ ever has its endpoints on 
			two different maximally connected pieces of the boundary by Observation~\ref{obs:independent}
			we can consider such maximally connected piece of the boundary of $H$ independently.
			For each of them embedding the threads in $T_H$ incident to vertices on that piece of boundary 
			is equivalent to embedding chords into a distorted cycle in a simple way which can be done in time polynomial in \(|T_H|\).
			We also remark, that one can assume that there is a cell \(\mathfrak{c}\) of the embedding of \(T_H\) inside \(H\) such that for each maximally connected piece of boundary of \(H\) either a thread incident to it or
			to a subset of it is present on the boundary of \(\mathfrak{c}\).
		\end{proof}}
	
	\both{
		
	For the simple embedding of \(T_H\) into \(H\) computed in Lemma~\ref{lem:stitches}, define the set of \emph{stitches} \(S_H\) of \(H\) as the planarization of the threads in this embedding.

}

	\section{The Patchwork Graph}
	\label{sec:patchworkdef}
	\latertitle{\section{Extended Version of Section~\ref{sec:patchworkdef}}}
	\both{After identifying a bounded number of stitches in each hole, we are finally able to define the \emph{patchwork graph} and prove desirable properties which we will use in our final application of Courcelle's theorem.
	An illustration of the patchwork graph is provided in Fig.~\ref{fig:pwgraph}.
	The following definition also doubles as a description of how to construct the patchwork graph from a given drawing.
	We remark that, unlike $\drawing$, the patchwork graph might be disconnected.
	
	\begin{figure}
		\centering
		\includegraphics[page=4]{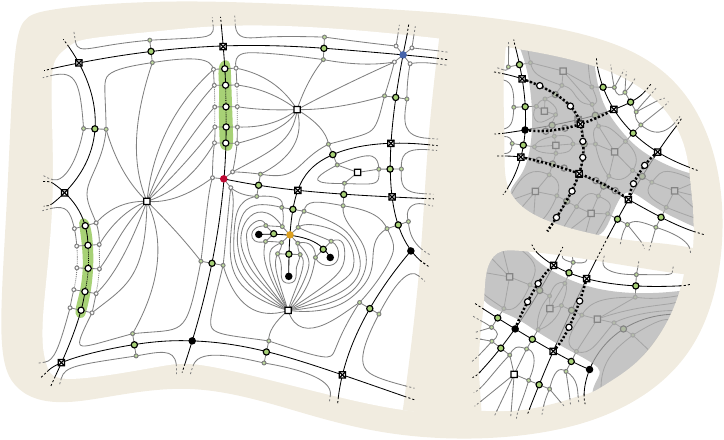}
		\caption{Illustration of a patchwork graph $\patchwork$. 
			The remainder of \patchwork is hinted in beige. 
			Black disks are original vertices.
			Colored disks are endpoints of edges in $\newedges$.
			Crossing vertices are crosses. 
			Green and white disks represent the edge segment/shadow vertices. 
			Cell vertices are white squares.
			Holes are shaded in gray and stitches drawn with thick, dashed curves.
		}
		\label{fig:pwgraph}
	\end{figure}
	
	\begin{definition}
	\label{def:patchwork}
	The \emph{patchwork graph} \(\patchwork\) and its embedding \(\mathcal{\patchwork}\) are given by the labeled graph derived from \drawing in the following steps:
		\begin{enumerate}
			\item Planarize \drawing and label the vertices which are newly introduced by this as \emph{crossing vertices}.
			Label vertices which correspond to vertices of \(G\) as \emph{real vertices}.
			Additionally label each \(s_i\) and \(t_i\) with label \(i \in [k]\).
			\item Subdivide each edge \(e\) in the planarization \planarization{\drawing} of \drawing by \(k\) vertices\footnote{
				If \(k = 1\) we subdivide by \(2\) vertices for reasons that will become clear when we introduce \emph{tracking labels}.}
			\(v^e_1, \dotsc, v^e_k\) which are labeled as \emph{segment vertices}---each segment vertex of \(e\) will represent a possible crossing point of the drawing of one of the \(k\) edges in \newedges and \(e\).
		\item Inside each face \(f\) of \planarization{\drawing}, introduce a new vertex \(v_f\) and label it as \emph{cell vertex}.
			\item Inside each face \(f\)  of \planarization{\drawing}, trace the boundary of \(f\) creating a curve at \(\varepsilon\)-distance and create a vertex labeled as \emph{shadow vertex} on this curve every time an endpoint of an edge in \newedges or a segment vertex is encountered.
			Insert two edges for each shadow vertex; one
			connecting the shadow vertex to the corresponding endpoint of an edge in \newedges or segment vertex; and one connecting the shadow vertex to \(v_f\).
			Note that 
			multiple shadow vertices can be introduced for the same vertex in \(G\) (e.g. the orange vertex in Fig.~\ref{fig:pwgraph}).
				Shadow vertices allow to distinguish different ways, more formally positions in the rotation around an endpoint, of accessing that endpoint via the inserted drawing of an edge in F; this is where the connectivity of \drawing is used (see Fig. 5).
%
			In this way each shadow vertex of an endpoint corresponds to an \emph{access direction}.
			\item Delete every vertex that is in the interior of a hole \(H\).
			\item For each hole \(H\) insert all stitches \(S_H\) for \(H\) into the interior of \(H\) and label the inserted vertices as \emph{crossing vertices}.\appendixandlong{\footnote{This means they receive the same label as vertices introduced by planarizing \drawing.}}
			\item For technical reasons which will become apparent later (when we introduce \emph{tracking labels}), we replace each edge in \(S_H\) by a path consisting of two vertices and three edges and label the inserted vertices as \emph{segment vertices}.\appendixandlong{\footnote{This means they receive the same label as vertices introduced by subdividing edge segments of the planaritation of \drawing.}}
		\end{enumerate}
	\end{definition}

	\begin{figure}
		\centering
		\includegraphics[page=5]{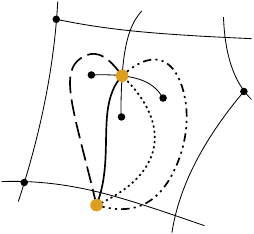}
		\caption{Illustration for different access directions. Each hypothetical drawing (indicated as thick dashes, normal, dotted, and dash-dotted lines) of the added edge between the orange vertices crosses the same edge segment of \(\planarization{\drawing}\) but separates the black vertices differently.
		In connected initial drawings, ways of separating vertices of the same cell by the drawing of an added edge are completely determined by potential crossing points of that drawing and its positions in the rotations around each of its endpoints.
		This is not the case for disconnected initial drawings.}
		\label{fig:shadowexpl}
	\end{figure}
	
	We introduce additional crossability labels for segment vertices in the following way.
	For every segment vertex $v$ corresponding to an edge segment $\sigma$ of edge $e \in E(G)$ 
	we label $v$ as \emph{crossable} for some edge $s_it_i \in F$ if either
	$e$ is not $H$-torn for any hole $H$, or
for \emph{each} hole \(H\) in \drawing for which \(e\) is \(H\)-torn, $\sigma$ lies on a part (when considering parts arising from the removal of the interior of \(H\)) of $e$ that is crossable for $s_it_i$.
}

	\both{%
	\begin{lemma}[$\star$]	
		\label{lem:crossable-label}
		If there exists a solution for the considered \SDE instance, then there is a solution such that all segment vertices which correspond to edge segments of an edge that is crossed by the drawing of \(s_it_i \in \newedges\) in the solution are labeled as crossable for \(s_it_i\).
	\end{lemma}}
	\later{%
	\begin{proof}
		By Lemma~\ref{lem:unnecessary} we only need to consider 
		solutions without removable detours.
	
		Assume for contradiction that there is a solution without removable detours and
		the drawing of $s_it_i$ crosses an edge segment $\sigma$ which is not labeled as crossable for \(s_it_i\).
		Then $\sigma$ lies on an edge part \(e\) of an \(H\)-torn edge which is not crossable for \(s_it_i\) for some hole \(H\) in \drawing.
		Since we consider a solution without removable detours
		the drawing of \(s_it_i\) in the solution corresponds to a solution curve for \(s_it_i\) that crosses the edge part \(e\), 
		contradicting the fact that \(e\) is not crossable.
	\end{proof}}

	\both{
	Note that Lemmas~\ref{lem:findcrossable} and~\ref{lem:stitches} and 
	Definition~\ref{def:patchwork} allow us to compute the patchwork graph in \FPT\ time.
%
	The two most important properties of the patchwork graph are encapsulated in the following two lemmas. 
	The proof of the first essentially relies on obtaining a bound on the diameter of each connected component of the patchwork graph---a task which is intuitively clear, but requires us to overcome some technical challenges due to the addition of stitches.}

	\both{%
	\begin{lemma}[$\star$]
		\label{lem:patchwork}
		The patchwork graph \patchwork has treewidth bounded by 
		\(3(2+4(k-1))(4\ell + 8(kf(k,\ell) - 1))\),
		where \(f(k,\ell)\) is the bound 
		on the number of crossable edge parts for a single added edge and hole obtained in Lemma~\ref{lem:crossable}.
	\end{lemma}}
	
	\later{
	\begin{proof}
		We show that \patchwork is a subgraph of a planar graph with diameter bounded in \((2+4(k-1))(7\ell + 4kf(k,\ell))\).
		Then the claim follows from \cite{RobertsonS84}.
		Planarity of \patchwork itself follows directly from the planarity of $\planarization{\drawing}$ and
		the way we constructed \patchwork from $\planarization{\drawing}$ as described in Definition~\ref{def:patchwork}.
		We define a supergraph \(\patchwork^+\) of \(\patchwork\) as follows.
		Add a vertex into each cell created in a hole $H$ when inserting the stitches $S_H$ and 
		connect if to the vertices on the boundary of that cell.
		Clearly, \(\patchwork^+\) is planar and,
		moreover, $\patchwork^+$ is connected.
		
		We first show that the diameter of \(\patchwork^+\) is at most \((2+4(k-1))(4\ell + d^* + 4)\) 
		where \(d^*\) is an upper bound for the distance between the 
		vertices on the boundary of any hole \(H\) to vertices in \(H\), and then proceed to show that \(d^* \leq 8(kf(k,\ell) - 1)\).
		
		For the first step we proceed inductively.
		Let $\patchwork_{k'}$ be constructed analogously to the patchwork graph \patchwork, but considering only the set $F_{k'} = F \setminus \{s_it_i \mid i > k'\}$ as the set of added edges for the definition of holes, and without adding threads and stitches, i.e., we skip steps six and seven in Definition~\ref{def:patchwork}.
		For example, $P_1$ is just the patchwork graph constructed from $\drawing$
		assuming there is only one added edge $s_1t_1$ and each hole in $P_1$ is empty.
		Observe that each hole of \(\patchwork_{k'}\) is a subset of a hole in $\patchwork_{k'-1}$.
		
		Instead of directly bounding the diameter of $\patchwork^+$ we prove the following,
		slightly stronger, claim.
		\begin{newcounterclaim}
			\label{clm:bddiameter}
			Any connected graph $\patchwork_{k'}^+$ that arises from \(\patchwork_{k'}\) by inserting an arbitrary graph of bounded diameter \(d\) into each hole of \(\patchwork_{k'}\) has diameter at most \((2+4(k'-1))(4\max_{i \in [k']}\ell_i + d + 4)\) for every $k' \leq k$ where
			$d$ is the maximum diameter of any plane graph inserted into a hole of $\patchwork_{k'}$ to construct $\patchwork_{k'}^+$.
			Moreover there is a \(u\)-\(v\)-path of at most this length for arbitrary \(u\) and \(v\) in \(\patchwork_{k'}^+\) that can be subdivided into at most three subpaths,
			\begin{itemize}
				\item the first of which is entirely inside the hole of \(\patchwork_{k'}\) inside which \(u\) is (if \(u\) is inside a hole),
				\item the second of which does not enter any hole of \(\patchwork_{k'}\),
				\item and the third of which is entirely inside the hole of \(\patchwork_{k'}\) inside which \(v\) is (if \(v\) is inside a hole).
			\end{itemize}
		\end{newcounterclaim}
		\begin{proof}
			We show the claim by induction on $k' \leq k$.
			For \(k' = 1\), by construction any cell vertex is at distance at most \(4\ell_1 + 2\) from \(s_1\) in \(\patchwork_{1}^+\).
			Furthermore, any vertex of \(\patchwork_{1}^+\) is by construction at distance at most \(d + 2\) from a cell vertex;
			for vertices outside holes, the distance is easily seen to be at most two to the cell vertex 
			of the cell of \drawing containing that vertex;
			for vertices inside holes, the distance is at most \(d + 2\) to the closest cell vertex 
			of a cell that shares some part of its boundary with the respective hole.
			This completes the base case, as any pair of vertices is connected via a path with the desired properties
			of length at most \(2(4\ell_1 + d + 4)\) using \(s_1\) and a cell vertices closest to each of the vertices in the pair.
			
			Now assume that any connected graph $\patchwork_{k'-1}^+$ that arises from \(\patchwork_{k' - 1}\) by inserting an arbitrary graph of bounded diameter \(d'\) into each hole of \(\patchwork_{k' - 1}\) has diameter at most \((2+4(k'-2))(4\max_{i \in [k' - 1]}\ell_i + d' + 4)\),
			and that for every \(u,v \in V(\patchwork_{k' - 1}^+)\) there is a \(u\)-\(v\)-path of this length in \(\patchwork_{k' - 1}^+\) that can be subdivided into at most three subpaths,
			\begin{itemize}
				\item the first of which is entirely inside the hole of \(\patchwork_{k' - 1}\) inside which \(u\) is (if \(u\) is inside a hole),
				\item the second of which does not enter any hole of \(\patchwork_{k' - 1}\),
				\item and the third of which is entirely inside the hole of \(\patchwork_{k' - 1}\) inside which \(v\) is (if \(v\) is inside a hole).
			\end{itemize}
			Consider the overlay of $\patchwork_{k'-1}$ and $\patchwork_{k'}^+$ in the plane and identify 
			vertices and edges that are in both graphs.
			Now, all vertices in \(V(\patchwork_{k'}^+) \setminus V(\patchwork_{k' - 1})\) lie inside holes of \(\patchwork_{k' - 1}\).
			For any pair of vertices of \(\patchwork_{k'}^+\) inside holes of \(\patchwork_{k' - 1}\) which are not holes in \(\patchwork_{k'}\), we can argue in a similar way as in the base case that they are connected by a path of length at most \(2(4\ell_{k'} + d + 4)\) in \(\patchwork_{k'}^+\).
			In particular, any vertex of $\patchwork_{k'}^+$
			that lies inside such a hole of $\patchwork_{k'-1}$ has distance at most \(2(4\ell_{k'} + d + 4)\)
			to the boundary of that hole.
			It remains to consider vertices of \(\patchwork_{k'}^+\) which are in \(\patchwork_{k' - 1}\) or inside holes of \(\patchwork_{k' - 1}\) which coincide with holes of \(\patchwork_{k'}\).
			
			Here we invoke the induction hypothesis on \(\patchwork_{k' - 1}^+\) where the graph inserted into a hole \(H\) of \(\patchwork_{k' - 1}\) is
			\begin{itemize}
				\item the same as the ones inserted into the respective hole of \(\patchwork_{k'}\) in \(\patchwork_{k'}^+\), if \(H\) is also a hole for \(\patchwork_{k'}\); and
				\item a star, which connects to all vertices of \(\patchwork_{k' - 1}\) on the boundary of \(H\) otherwise.
			\end{itemize}
			From this we get that \(\patchwork_{k'-1}^+\) has diameter at most \((2+4(k'-2))(4\max_{i \in [k' - 1]}\ell_i + \max\{2,d\} + 4)\),
			and that for every \(u,v \in V(\patchwork_{k' - 1}^+)\) there is a \(u\)-\(v\)-path of this length in \(\patchwork_{k' - 1}^+\) that can be subdivided into at most three subpaths,
			\begin{itemize}
				\item the first of which is entirely inside the hole of \(\patchwork_{k' - 1}\) inside which \(u\) is (if \(u\) is inside a hole),
				\item the second of which does not enter any hole of \(\patchwork_{k' - 1}\),
				\item and the third of which is entirely inside the hole of \(\patchwork_{k' - 1}\) inside which \(v\) is (if \(v\) is inside a hole).
			\end{itemize}
			Hence vertices in \(\patchwork_{k'}^+\) which are also in \(\patchwork_{k' - 1}\) or inside holes of \(\patchwork_{k' - 1}\) which coincide with holes of \(\patchwork_{k'}\) have distance at most \((2+4(k'-2))(4\max_{i \in [k' - 1]}\ell_i + \max\{2,d\} + 4)\) in \(\patchwork_{k'}^+\) because by there are paths witnessing this that do not use the stars inserted into holes of \(\patchwork_{k' - 1}\) in \(\patchwork_{k' - 1}^+\) that do not coincide with holes in \(\patchwork_{k'}\) and hence correspond to paths in \(\patchwork_{k'}^+\).
			
			All together the described paths can be combined to form paths in \(\patchwork_{k'}^+\) with the desired structural properties and of length at most \((2+4(k'-1))(4\ell + d + 4)\) between any pair of vertices in \(\patchwork_{k'}^+\).
		\end{proof}
	
		Next we show that \(d^* \leq 8(kf(k,\ell) - 1)\).
		Recall that vertices of \(\patchwork^+\) inside $H$ are either vertices on the boundary of \(H\), crossing vertices introduced when inserting the stitches \(S_H\), or 
	 	newly inserted cell vertices inside \(H\) that were created when constructing $\patchwork^+$ from $\patchwork$.
		
		\begin{newcounterclaim}
			\label{clm:dstar}
			\(d^*\) is at most \(8(\max_{H\text{ hole in \drawing}}|T_H|-1) \leq 8(kf(k,\ell) - 1)\).
		\end{newcounterclaim}
%
		\begin{proof}
			Consider an arbitrary hole \(H\) in \drawing and assume at first that the boundary is only one connected piece.
			Then, the distance between a crossing vertex $v$ in \(H\) and 
			a cell vertex in a cell incident to some piece of the boundary of $H$ and 
			a thread containing $v$ is at most $3(|T_H| - 2)$. 
			This is, since such a thread is crossed by at most \(|T_H| - 2\) threads which 
			\(v\) does not lie on and 
			between every two crossings of threads we introduced two vertices.
			Furthermore, the distance between two arbitrary cell vertices  incident to the same boundary piece of $H$ is at most $2|T_H|$.
			This is, because two such cell vertices are separated by at most all threads in $T_H$.
			Consequently, one needs to traverse at most one segment or crossing vertex per thread and
			one cell vertex per cell incident to the boundary of $H$.
			Using such paths together with the at most one edge to connect segment and 
			cell vertices to a closest crossing vertex 
			we can connect any pair of vertices of \patchwork inside \(H\) by paths of length at most \(8|T_H| - 10\).

			Finally, if the boundary of $H$ consists of multiple pieces,
			recall that the threads incident to different pieces of the boundary of $H$ do not intersect (Observation~\ref{obs:independent}).
			Consequently, the derived stitches do not share a crossing or segment vertex.
			Moreover, by the construction in Lemma~\ref{lem:stitches} there is one cell incident to a crossing or segment vertex of some stitch for each piece of the boundary.
			This means, in $\patchwork^+$  there exists a cell vertex $c$ that is adjacent to all of them.
			Reaching this vertex $c$ from a cell vertex in some cell incident to some boundary piece of $H$ 
			needs at most a path of length $2|T_H| + 1$.
			The argument is as above, we at most cross every thread and in the end have to go to $c$ itself.
			Hence, $c$ can be reached from any cell vertex of a cell incident to the boundary of $H$ in at most $2|T_H| + 1$ steps.
			This gives the final bound of $6(|T_H| - 2) + 2(|T_H| + 1) + 2 = 8(|T_H| - 1)$ to connect any pair of vertices inside $H$.
		\end{proof}
	
		Now, the lemma follows by applying Claim~\ref{clm:bddiameter} with $k' = k$ and 
		the bound on $d^*$ obtained in Claim~\ref{clm:dstar}.
	\end{proof}}

	\both{The second lemma will later allow us to check whether two edge segments in $\patchwork$ belong to the same edge in $\drawing$ via an MSO formulation.}

	\both{%
	\begin{lemma}[$\star$]	
		\label{lem:crossable-conncetion}	
		Segment vertices which correspond to edge segments of the same edge in \(e \in E(G)\) and are labeled as crossable for \(s_it_i\) 
		are connected via paths in \patchwork consisting only of segment and crossing vertices which 
		correspond to segments and crossings of \(e\) and segments and crossings for threads that connect parts of \(e\).
	\end{lemma}}
	\later{%
	\begin{proof}
		Segment vertices which correspond to segments of an edge \(e\) which is not \(H\)-torn for any hole \(H\) are obviously connected via a path in \patchwork that consists only of segment vertices and crossing vertices following \(e\).
		
		For an edge \(e \in E(G)\) which is \(H\)-torn for some hole \(H\) in \drawing we argue as follows:
		Consider two segments of \(e\) for which the corresponding segment vertices \(v_1\) and \(v_2\) are both marked as crossable.
		Since these segments both belong to \(e\), they are connected along a traversal of \(e\) in \drawing.
		We fix this traversal of \(e\).
		The traversal of \(e\) gives rise to a sequence of walks in \patchwork
		consisting of segment and crossing vertices arising from \(e\) which are interrupted by holes.
		We proceed by induction on the number \(w\) of such maximal walks the traversal of \(e\) gives rise to.
		
		For the base case, assume that the traversal of \(e\) gives rise to \(w = 2\) maximal walks, one of which contains \(v_1\) and the other of which contains \(v_2\).
		These walks can only be interrupted by a single hole \(H\).
		Because \(v_1\) and \(v_2\) are labeled as crossable for \(s_it_i\), 
		both these walks correspond to curves which are subsets of edge parts of the \(H\)-torn edge \(e\) 
		and both are crossable for \(s_it_i\).
		Moreover, since the traversal of \(e\) gives rise to only these two maximal walks, 
		there is no other edge part of \(e\) between these edge parts along the traversal of \(e\).
		This means a thread is inserted between these parts of \(e\) which then leads to a connection via stitches.
		
		Now assume that the claim holds in case the number of maximal walks the traversal of \(e\) gives rise to is \(w\), and consider the case that the number of maximal walks the traversal of \(e\) gives rise to is \(w + 1\).
		Consider the walk \(W_1\) that contains \(v_1\) and the walk \(W_1'\) that succeeds this walk along the traversal of \(e\).
		If \(W_1'\) contains a segment vertex \(v'\) that is marked as crossable for \(s_it_i\), we can apply the base case to \(v_1\) and \(v'\), and the induction hypothesis to \(v'\) and \(v_2\) to find a desired path in \patchwork.
		
		Otherwise, there is no such segment vertex on \(W_1'\) that is marked crossable for \(s_it_i\).
		Hence, there is some hole $H'$ in \drawing such that $e$ is $H'$-torn and 
		\(W_1'\) corresponds to a curve which is a subset of an edge part of \(e\) with respect to \(H'\) which 
		is not crossable for \(s_it_i\).
		Recall that \(v_1 \in W_1\) is labeled as crossable for \(s_it_i\). 
		Consequently, \(W_1\) corresponds to a curve which for each hole \(H\) in \drawing,
		such that $e$ is $H$-torn,
		is a subset of an edge part of \(e\) with respect to \(H\) which is crossable for \(s_it_i\).
		This leads to the observation that \(H'\) is the hole in \drawing that interrupts \(W_1\) from \(W_1'\),
		as else $W_1$ and $W_1'$ correspond to subsets of the same edge part with respect to \(H'\).
		Now consider \(v_2\).
		First, we observe that it is also contained in an edge part of the \(H'\)-torn edge \(e\).
		Second, as \(v_2\) is labeled `crossable' for \(s_it_i\) this edge part is also crossable for \(s_it_i\) and $H'$.
		Hence, the earliest edge part of the \(H'\)-torn edge \(e\) that is crossable for \(s_it_i\) 
		that follows $W_1$ in the traversal of $e$ from $v_1$ to $v_2$ is well-defined.
		Moreover, this part contains one of the considered walks after \(W_1\) which 
		is connected to \(W_1\) via stitches and
		contains a segment vertex \(v\) that is labeled as crossable for \(s_it_i\).
		We complete the proof by invoking the induction hypothesis for \(v\) and \(v_2\).
	\end{proof}}


	\both{Intuitively, we would like Lemma~\ref{lem:crossable-conncetion} to lead to an MSO subformula that can check whether two segment vertices in $\patchwork$ belong to the same edge---an important component of our algorithm for \SDE. The lemma provides us with a characterization that seems suitable for this task since it is easy to define a path in MSO, but there is an issue if we use $\patchwork$ as it is currently defined: a crossing vertex is adjacent to $4$ segment vertices, and $\patchwork$ (viewed as a graph without an embedding)}
	\later{ currently }%
	\both{does not specify which of these segment vertices belong to the same edge.}
	\short{%
	We resolve this by introducing \emph{tracking labels}: for each crossing vertex $v$ in \patchwork created by a crossing between edges $e$ and $e'$ in $\drawing$, we assign the label $1$ to the two unique neighbors of $v$ corresponding to $e$ and the label $2$ to the remaining two neighbors of $v$.}
	
	\later{
		
	We resolve this by introducing \emph{tracking labels} in the following way:
	For each crossing vertex $v$ in \patchwork created by a crossing between edges $e$ and $e'$ in $\drawing$, we assign the label $1$ to the two unique neighbors of $v$ in \patchwork that are segment vertices corresponding to segments of the edge $e$ or segments of threads connecting parts of \(e\), and assign the label $2$ to the remaining two neighbors of $v$ in \patchwork (which must be segment vertices corresponding to segments of $e'$ or segments of threads connecting parts of \(e\)). We break the symmetry between $e$ and $e'$ arbitrarily\footnote{This symmetry breaking is unproblematic because we subdivide each edge segment at least twice; no crossing vertices share neighbors.}.
	Then from the definition of the tracking labels and Lemma~\ref{lem:crossable-conncetion}, we get:}

	\both{\begin{corollary}
	Segment vertices which correspond to edge segments of the same edge in \(e \in E(G)\) and are labeled as crossable for \(s_it_i\) are connected via paths in \patchwork consisting only of segment and crossing vertices with the following property: the two neighbors of each crossing vertex on the path are segment vertices with the same tracking label.
	\end{corollary}}

\section{Using the Patchwork Graph}
\label{sec:MSO}
\latertitle{\section{Extended Version of Section~\ref{sec:MSO}}}

\both{Now that we have constructed the patchwork graph $\patchwork$ and established that it has the properties we need, we can proceed to the final stage of our proof. Here, our aim will be to identify a combinatorial characterization which projects the behavior of a solution from $\drawing$ to $\patchwork$, establish a procedure that allows us to identify (and construct) solutions based on a characterization in $\patchwork$, and finally show how to find such characterizations. To streamline our presentation, at this stage we perform a brute-force branching procedure which will determine, for each $s_it_i\in F$, the number $\ell'_i$ of crossings between the curve connecting $s_i$ to $t_i$ and edges of $\drawing$ in the sought-after solution.} 
\later{Clearly, $\ell'_i\in [\ell_i]$ and this branching procedure only incurs a multiplicative runtime cost of at most $\ell^k$.}

\both{Consider a hypothetical solution $S$,
and let $f$ be a curve in $S$ connecting vertex $a$ to $b$.
The \emph{trace} $r_f$ of $f$ is a walk in \patchwork starting at $a$ such that:
\begin{enumerate}
	\item From \(a\), \(r_f\) proceeds to the shadow vertex that corresponds to the access direction through which $f$ connects to $a$, and then to the cell vertex of the first cell $\mathfrak{c}_1$ in $\drawing$ intersecting $f$.
	\item For each intersection along $f$ with an edge segment \(q\) between cells \(\mathfrak{c}_i\) and \(\mathfrak{c}_{i + 1}\), \(r_f\) proceeds to the shadow vertex of a segment vertex \(v\) in \(\mathfrak{c}_i\) on \(q\), then to \(v\), then to its shadow vertex in \(\mathfrak{c}_{i + 1}\), and then to the cell vertex of \(\mathfrak{c}_{i + 1}\), where \(v\) has the property that the number of segment vertices of \(q\) on either side of \(v\) are at least as large as the number of drawings of added edges in \newedges which intersect \(q\) on the respective side of its intersection with \(f\).
		Such a segment vertex \(v\) exists, since there are \(k = |\newedges|\) segment vertices on \(q\).
	\item Finally, \(r_f\) continues to the shadow vertex that \emph{corresponds} to the direction through which $f$ enters $b$, and finally ends in $b$.
\end{enumerate}

Observe that $r_f$ visits precisely $4\ell'_i + 5$ vertices.
Moreover, for two curves $f,f'$ in $S$, their traces $r_f,r_{f'}$ may only intersect in cell vertices, the real vertices that form the endpoints of the curves, and the associated shadow vertices.

Now, let the \emph{solution trace} ($r_S$,$\eta_S$) of $S$ be a pair where $r_S=\{r_f|f\in S\}$ and $\eta_S$ describes cyclic orders
which will intuitively capture how edges cross into and out of each cell vertex in the solution. Let $R_S=\{v~|~ \exists f\in S: v\in r_f\}$ be the set of all vertices occurring in the traces of $S$. $\eta_S$ then is a mapping from each cell vertex $c\in R_S$ to a cyclic order $\prec_c$ over the shadow vertices in $R_S$ that are incident to $c$.
Specifically, \(\prec_c\) is defined as the cyclic order given by the cycle on the neighborhood of \(c\) in \patchwork restricted to \(R_S\).

Solution traces describe the way in which a solution can be related to a set of walks and cyclic orders in \patchwork.
	Of course we can abstract away from the explicit reference to a solution and define the more general notion of \emph{preimages} whose combinatorial structure is the same as that of a solution trace but which does not arise and in particular does not even need to correspond to a solution.
	(Preimages and solution traces relate in a similar way as solution curves and solutions in Section~\ref{sec:stitches}.)}
	\short{%
Moreover, it is possible to devise a polynomial-time \emph{assembly procedure} $\mathbf{A}$ which ``reverse-engineers'' the definition of solution traces and can construct a set of curves from a preimage (note, however, that not every preimage will lead to an actual solution). In principle, we could at this point solve \SDE\ by running $\mathbf{A}$ on every preimage in $\patchwork$, but that is not feasible simply because there may be too many preimages in $\patchwork$. However, when one abstracts away the specific way a preimage is embedded in $\patchwork$ (intuitively, by forgetting the unique identities of vertices in $\patchwork$ and only keeping their labels), there is only a bounded number of distinct possibilities of how a preimage may look like. This is formalized via the notion of \emph{template traces}~\textbf{(\(\star\))}.}

\later{%
Formally, a \emph{preimage} \((\alpha',\beta')\) is a tuple with the following properties.
$\alpha'$ is a set of $k$ walks in $H$ which are labeled $\alpha'_1,\dots,\alpha'_k$, where each $\alpha'_i$ has length $4(\ell_i + 1)$ and visits vertices with the same orders of labels as traces.
Similarly, $\beta'$ is a mapping from each cell vertex \(c\) visited by the walks in $\alpha'$ to the cyclic order over its neighbors that occur in $\alpha'$, along the cycle on \(N_\patchwork(c)\) in \patchwork.

Obviously every solution trace is a preimage.
	Conversely, one can derive a drawing of all edges of \(F\) into \drawing from a preimage \((\alpha',\beta')\) by the \emph{assembly procedure} $\mathbf{A}$ introduced below.
	For each $\alpha'_i\in \alpha'$, $\mathbf{A}$ will draw a curve $u_i$ that starts and ends at the two vertices labeled $i$ (i.e., the endpoints of $s_it_i\in F$) as described in the following steps.
\begin{enumerate}
	\item $u_i$ exits its starting vertex via the access direction given by the first shadow vertex in $\alpha'_i$.
	\item For each cell vertex $c$ such that $(e_1,v_1,c,v_2,e_2)$ forms a subsequence of visited vertices in $\alpha'_i$, expand $u_i$ by drawing a curve $\iota$ in $c$ connecting the edge segment (or the real vertex) $e_1$ to the edge segment (or the real vertex) $e_2$ in the following way.
	\begin{itemize}
		\item Consider an arbitrary other curve drawn in $c$ by $\mathbf{A}$ up to now, say $\zeta$, that was obtained from some subsequence $(e_1^\zeta,v_1^\zeta,c,v_2^\zeta,e_2^\zeta)$. $\iota$ will intersect $\zeta$ if and only if the shadow vertices of $\iota$ interleave with the shadow vertices of $\zeta$ in $\beta(c)$ (i.e., for instance, if $v_1\prec_c v_1^\zeta \prec_c v_2 \prec_c v_2^\zeta \prec_c v_1$).
		\item Such a drawing can be achieved by, e.g., having the curve $\iota$ follow the inside boundary of $c$ in a clockwise manner while avoiding all curves it is not supposed to cross (as these will be either completely enveloped by or completely enveloping $\iota$).
		\item We remark that $v_1$ and $v_2$ may either be shadows of segment vertices or the actual endpoints $s_i$ or \(t_i\).
	\end{itemize}
	\item $u_i$ ends by entering the final real vertex in $\alpha'_i$ from the direction specified by the last shadow vertex in $\alpha'_i$.
\end{enumerate}
The intuition here is that \(\textbf{A}\) interprets a preimage of a template trace as a specification of precisely which parts of $\drawing$ should be crossed by the drawings of each added edge (this information is provided in $\alpha'$), while controlling when and how individual curves in the newly constructed solutions should cross each other (this information is provided in $\beta'$).
	Note that the output of \(\textbf{A}\) for an arbitrary preimage will in general not be a solution for our edge insertion problem, but---crucially---one can check whether it is in polynomial time.

Observe that, although preimages imply curves in \drawing for all added edges in \newedges, and we can check for each of them if they are a solution, we cannot iterate over them in \FPT\ time as the number of preimages in \patchwork is generally not \FPT.
	We will however be able to distill the structure of preimages, independently of their exact specification in \patchwork.
	For this we define \emph{template traces}.
%
A template trace is a tuple $\tau=(T,\alpha,\beta)$ where:
\begin{itemize}
	\item $T$ is a graph whose vertices are equipped with a labeling that matches the vertex-labeling used in \patchwork (i.e., some may be labeled as segment vertices, some as cell vertices, etc., and in addition some of them may be labeled as the endpoints of added edges in $F$);
	\item $\alpha=\{\alpha_1,\dots,\alpha_k\}$ is a set of walks in $T$, where each walk $\alpha_i$ has length $4(\ell'_i + 1)$ and the types of vertices visited by $\alpha_i$ match the types of vertices visited by a trace (i.e., $\alpha_i$ starts with a real vertex labeled $i$, then proceeds with a shadow vertex, a cell vertex, followed by a sequence of $\ell_i'$-many subsequences of shadow-, segment-, shadow-, cell vertices, and ends with a shadow vertex followed by a different real vertex labeled $i$); and
	\item $\beta$ is a mapping from each cell vertex in $T$ to a cyclic order over its adjacent shadow vertices.
	\item For simplicity, we require that each vertex and edge in $T$ occurs in at least one walk in $\alpha$.
\end{itemize}}


\both{%
\begin{proposition}[$\star$]
	\label{prop:constructtraces}
	There are at most $(k\ell)^{\bigoh(k\ell)}$ distinct template traces. Moreover, the set of all template traces can be enumerated in time $(k\ell)^{\bigoh(k\ell)}$.
\end{proposition}}

\later{%
\begin{proof}
	First, observe that $T$ contains at most $k \cdot (4\ell + 5)$-many vertices. The number of walks of length at most $4(\ell + 1)$ over the vertex set of $T$ can hence be upper-bounded by $(k\ell)^{\ell}$. As for $\beta$, observe that each shadow vertex in $T$ is only adjacent to a single cell vertex, and hence occurs in at most one cyclic order in the image of $\beta$. Hence, once we fix one possible choice for $\alpha$, the number of all possible $\beta$'s is upper-bounded by the number of ways of partitioning all (at most $k\ell$) shadow vertices in $T$ into at most \(k\ell\) parts and then considering all possible permutations of the obtained parts. To show that this is also upper-bounded by $(k\ell)^{\bigoh(k\ell)}$, observe that the desired number of partitionings with internal permutations is upper-bounded by the number of permutations over the set of all shadow vertices and an equal-cardinality set of auxiliary ``separating vertices'', whose sole role is to model possible partitionings of shadow vertices. The number of such permutations is, naturally, in $(k\ell)^{\bigoh(\ell)}$. This also yields a procedure that can construct the set of all possible template traces.
\end{proof}}

\later{
	
	We say that a template trace \((T,\alpha,\beta)\) \emph{matches} a preimage \((\alpha',\beta')\) if there is a label-preserving bijective mapping $\gamma$ (called the \emph{preimaging}) from the vertices on walks in \(\alpha'\) to $V(T)$ such that (1) for each $\alpha'_i\in \alpha'$, $\gamma(\alpha'_i)=\alpha_i$ and (2) $\gamma(\beta')$ maps each \(c\) to \(\beta(\gamma(c))\).
	For a template trace \(\tau\) that matches a preimage \((\alpha',\beta')\), we say that \((\alpha',\beta')\) is a \emph{preimage} of \(\tau\).
	Intuitively, a preimage of a template trace is its firmly embedded counterpart in \patchwork.
	As every solution trace is a preimage, these definitions carry over to solution traces.
}

\short{%
Every preimage gives rise to a template trace; in this case we say that the template trace \emph{matches} the preimage, or the preimage is a preimage \emph{of} the template trace.
While a template trace can match multiple preimages, the following lemma shows that a template trace contains sufficient information to \emph{almost} reconstruct a solution using \textbf{A} on any preimage of a template trace matching a hypothetical solution trace.

\begin{lemma}[\(\star\)]
	\label{lem:usetraces}
	Let $S$ be a solution which matches a template trace $\tau = (T,\alpha,\beta)$, and let $(\alpha',\beta')$ be a preimage of $\tau$. Let $S'$ be the output of $\mathbf{A}$ applied to $(\alpha',\beta')$. 
	Then $S'$ is either a solution, or there exists an edge $e$ of $G$ that intersects some curve in $S'$ more than once.
\end{lemma}}

\later{
The following lemma shows that a template trace trace \(\tau\) matching the solution trace of a hypothetical solution contains a sufficient amount of information to \emph{almost} reconstruct a solution using \textbf{A} on a preimage of \(\tau\).

\begin{lemma}
	\label{lem:usetraceslong}
	Let $S$ be a solution which matches a template trace $\tau = (T,\alpha,\beta)$, and let $(\alpha',\beta')$ be a preimage of $\tau$. Let $S'$ be the output of $\mathbf{A}$ applied to $(\alpha',\beta')$. 
	Then $S'$ is either a solution, or there exists an edge $e$ of $G$ that intersects some curve in $S'$ more than once.
\end{lemma}
\begin{proof}
	The construction guarantees that for each added edge in $F$, say $s_it_i$, there will be a simple curve $u_i$ in $S'$ connecting $s_i$ to $t_i$. Moreover, $u_i$ will cross precisely as many curves in $\drawing$ as $S$, since $S$ matches $\tau$ and $(\alpha',\beta')$ is a preimage of $\tau$. To complete the proof, we argue that the crossings between distinct curves in $S$ are replicated by $S'$.
	
	Consider a crossing $x$ between the drawings of two added edges, say $s_i$-$t_i$ curve $p$ and $s_j$-$t_j$ curve $q$ in $S$, and assume this crossing occurs in a cell $\mathfrak{c}$ of \drawing. Let us split each such curve into \emph{curve segments}, which are maximal connected subcurves of these curves that do not intersect any other curves in $\drawing$ (in particular, each such curve segment must lie completely in a cell of $\drawing$, but a cell of \drawing may contain multiple curve segments of a single curve).
		\(x\) is contained in exactly one curve segment of \(p\) and one curve segment of \(q\), and both these segments start and end on the boundary of \(\mathfrak{c}\).
		Because the curve segments intersect exactly once in \(x\), their endpoints have to interleave on the boundary of \(\mathfrak{c}\).
		These start- and endpoints correspond via the preimaging between the solution trace of \(S\) and \(\tau\) to shadow vertices \(p_1\) and \(p_2\) and \(q_1\) and \(q_2\) respectively in \(V(T)\) which also interleave in the cyclic ordering given by \(\beta(c)\) where \(c\) corresponds to \(\mathfrak{c}\) via the preimaging between the solution trace and \(\tau\).
		This means, via the preimaging \(\gamma\) between \(\tau\) and \((\alpha',\beta')\), that \(\gamma(p_1)\), \(\gamma(c)\) and \(\gamma(p_2)\) are consecutive on \(\alpha'_i\), \(\gamma(q_1)\), \(\gamma(c)\) and \(\gamma(q_2)\) are consecutive on \(\alpha'_j\), and the pairs \((\gamma(p_1),\gamma(p_2))\) and \((\gamma(q_1),\gamma(q_2))\) interleave in \(\prec_{\gamma(c)}\).
		At this point, $\mathbf{A}$ ensures that the drawings of \(s_it_i\) and \(s_jt_j\) in \(S'\) intersect exactly once in the cell corresponding to \(\gamma(c)\)---mirroring the behavior of $S$.
	%
	
	On the other hand, let us now make an analogous argument to analyze what happens if two curves \(p\) (the drawing of \(s_it_i\)) and \(q\) (the drawing of \(s_jt_j\)) in \(S\) do not cross each other in a cell $\mathfrak{c}$ which they both intersect (drawings of added edges which do not intersect the same cell of \drawing in \(S\) are easily argued not to intersect the same cell of \drawing in \(S'\)).
	Consider two arbitrary curve segments of
	\(p\) and \(q\) in \(\mathfrak{c}\).
		Using the preimagings between the solution trace of \(S\) and \(\tau\) and \(\tau\) and \((\alpha',\beta')\) we can find shadow vertices in \patchwork that correspond to the endpoints of the curve segments and do not interleave in the cyclic ordering for the cell vertex corresponding to \(\mathfrak{c}\) in the preimage \((\alpha',\beta')\).
		As this is true for arbitrary edge segments of \(p\) and \(q\) in \(\mathfrak{c}\), applying \textbf{A} to \((\alpha',\beta')\) then ensures that the drawings of \(s_it_i\) and \(s_jt_j\) do not intersect in the cell that corresponds to the cell vertex that is associated to \(\mathfrak{c}\) via the preimagings.
	
	Altogether, we conclude that the curves obtained by $\mathbf{A}$ will contain the same number of crossings as the curves in $S$, and will pairwise cross each other if and only if they pairwise crossed each other in $S$ (and, in any case, at most once). Hence either $S'$ is a solution, or it contains a curve that crosses an existing edge in $\drawing$ more than once. 
\end{proof}}

\both{
	
	Next, we show that the problem of finding a preimage of a template trace (or determining that there is none) can be encoded in Monadic Second Order (MSO) logic.}

\short{%
\begin{lemma}[$\star$]
	\label{lem:MSOtrace}
	Let $\tau=(T,\alpha,\beta)$ be a template trace. There exists an MSO formula $\phi_\tau(V(T))$ of size independent of $G$ and $\drawing$ which is satisfiable in \patchwork if and only if there exists a preimage for $\tau$ in \patchwork.
\end{lemma}}

\later{%
\begin{lemma}
	\label{lem:MSOtracelong}
	Let $\tau=(T,\alpha,\beta)$ be a template trace. There exists an MSO formula $\phi_\tau(V(T))$ of size independent of $G$ and $\drawing$ which is satisfiable in \patchwork if and only if there exists a preimage for $\tau$ in \patchwork. Moreover, if the formula is true, then the interpretation of $V(T)$ defines a \emph{preimaging} between a preimage of \(\tau\) and $\tau$.
\end{lemma}

\begin{proof}
	We prove the lemma by giving the construction of $\phi_\tau(V(T))$. First of all, the formula requires that, for each walk $\alpha_i\in \alpha$, the variables in $\alpha_i\subseteq V(T)$ form a walk in \patchwork that visits the variables in the order prescribed by $\alpha_i$, whereas the interpretation of each variable has the correct label (including the specific labels marking the endpoints of walks and crossability of segment vertices). Observe that an interpretation satisfying this initial condition will result in a set of walks in \patchwork satisfying the requirements imposed on the first tuple of a preimage of $\tau$. 
	
	To complete the proof, we now need to ensure that the cyclic orders defined by $\beta$ (within the template trace $\tau$, over vertices in $T$) match those we obtain for the interpretation of $V(T)$ when following the procedure used to define cyclic orders for preimages. In particular, the following must hold: for each cell vertex $c'\in V(\patchwork)$ that is the interpretation of some cell vertex $c$ in $V(T)$, the cycle on the shadow vertices in the neighborhood of \(c'\) contains the vertices that are interpretations of shadow vertices in $V(T)$ in the same cyclic order as the one given in $\beta(c)$. To express this condition in MSO, we proceed as follows. For each pair of variables for shadow vertices, say $v$ and $v'$ in $V(T)$ that are adjacent to $c$ and directly consecutive in $\beta(c)$, we express the existence of an $v$-$v'$ path in \patchwork consisting exclusively of shadow vertices none of which are among the variables for vertices in \(V(T) \setminus \{v,v'\}\). It is easy to verify that the formula constructed in this way directly encodes all the requirements imposed on preimages of~$\tau$.
\end{proof}}

\both{
	
	Finally, we have all the ingredients needed to prove our main result.}

\both{%
\begin{theorem}[$\star$]
	\label{thm:main}
	\SDE\ is fixed-parameter tractable.
\end{theorem}}

\short{%
\begin{proof}[Proof Sketch]
	Given some instance $\mathcal I$ of \SDE,
	construct the patchwork graph \patchwork as described in Section~\ref{sec:patchworkdef}.
	Branch on the exact number of crossings between each added edge and \drawing and 
	the template trace matching a hypothetical solution as per Proposition~\ref{prop:constructtraces}.
	In each branch construct an MSO formula \(\psi_\tau\) using Lemma~\ref{lem:MSOtrace}, as well as Lemma~\ref{lem:crossable-conncetion} together with tracking labels to encode the fact that a preimage for the respective template trace can be embedded into \patchwork, in a way that the result of \textbf{A} applied to the preimage crosses only edge segments whose segment vertices are labeled as crossable for the respective added edge, and does not cross edge segments of the same edge in \drawing.
	Then employ Courcelle's theorem to decide if there is an interpretation of \(\psi_\tau\) in $\patchwork$.
	For each such interpretation apply \textbf{A} and check if the resulting drawing implies a solution for the \SDE instance.
	Reject if no branch encounters a solution.
	Correctness follows from Lemmas~\ref{lem:crossable-label} and \ref{lem:usetraces}.
\end{proof}}

\later{%
\begin{proof}
	We provide a fixed-parameter algorithm for \SDE. The algorithm begins by constructing the patchwork graph \patchwork as per Definition~\ref{def:patchwork}, including the crossability and other labels described in Section~\ref{sec:patchworkdef}.
	Then, we branch over the exact number of crossings of each added edge with \drawing and, for each branch we construct and branch over the set of all possible template traces as per Proposition~\ref{prop:constructtraces}. For each template trace $\tau$, we use Lemma~\ref{lem:MSOtracelong} to construct an MSO formula $\phi_\tau(\cdot)$ that checks for a preimage of $\tau$. We enhance $\phi_\tau(\cdot)$ with formulas which ensure that any preimage $(\alpha',\beta')$ obtained from its interpretation will satisfy the following conditions:
	(1) each walk \(\alpha_i \in \alpha'\) contains no segment vertex which is not labeled as crossable for \(s_it_i\);
		and (2) no two segment vertices occurring in the same walk in $\alpha'$ correspond to segments of the same edge in $G$.
		Encoding (1) is trivial.
		For (2) we can add formulas \(\phi^i_\tau(\cdot)\) for \(i \in [k]\) that exclude the existence of paths between segment vertices on \(\alpha'_i\) consisting only of segment and crossing vertices such that on these paths every time a crossing vertex succeeds a segment vertex \(v\) it is followed by a segment vertex \(v'\) such that \(v\) and \(v'\) have the same tracking label.
		We call the resulting formula \(\psi_\tau(\cdot)\).
	
	For each $\tau$, we now use Courcelle's theorem to compute an interpretation $\omega$ of $\psi_\tau(\cdot)$ or determine that there is none in $\patchwork$. For each interpretation computed in this way, we apply \textbf{A} and check if we get a solution to the original \SDE instance (before branching on the exact number of crossings for each added edge with \drawing). If we obtain a solution, we accept.
	
	If none of the formulas constructed in this way lead to a solution, then we reject. This is correct, as can be argued as follows.
	If the input instance $\mathcal{I}$ of \SDE has a solution, by Lemma~\ref{lem:crossable-label} there is a solution $S$ in which the drawing of \(s_it_i\) only crosses edge segments of which the corresponding segment vertices are labeled as crossable for \(s_it_i\).
	The solution trace of \(S\) must match some template trace $\tau'$. Moreover, the solution trace $(r_S,\eta_S)$ is a witness for the existence of a preimage of $\tau'$ which will also satisfy $\phi^i_{\tau'}(\cdot)$ (since we know that $S$ did not have any double-crossings). Hence the second case of Lemma~\ref{lem:usetraceslong} cannot occur, and we would have found an interpretation of $\psi_{\tau'}(\cdot)$ that leads to a solution.
\end{proof}}

\both{
	
	Theorem~\ref{thm:main} immediately implies the fixed-parameter tractability of \SCEI\ and \textsc{SLCEI} parameterized by $k+\ell$ (see Section~\ref{sec:prelims}). But the same approach can also be used to obtain the fixed-parameter tractability of the other problems of interest to us that were defined in the introduction, with only minor adaptations required.}

\both{%
\begin{theorem}[$\star$]
	\label{thm:therest}
	\textsc{S$\ell$-PEI}, $\ell$-\textsc{PEI} and \textsc{Locally Crossing-Minimal Edge Insertion} are fixed-parameter tractable when parameterized by $\ell+k$.
\end{theorem}}

\later{%
\begin{proof}
	For \textsc{S$\ell$-PEI}, we simply need to adapt the proof of Theorem~\ref{thm:main} in a way which ensures that after we add the solution curves to $\drawing$, every edge in $G$ still has at most $\ell$ crossings. To do this, first we'll enrich the labeling in $\patchwork$ as follows: for each segment vertex in $\patchwork$ which is labeled as crossable for any added edge in \newedges that correspond to an edge $e$ in $\drawing$, we add a ``\emph{crossings}'' label which contains the total number of crossings of $e$ in $\drawing$. Otherwise, we proceed exactly as in the proof of Theorem~\ref{thm:main}, up to the construction of the formula $\psi_\tau(\cdot)$. At that point, we enrich $\psi_\tau(\cdot)$ with an auxiliary MSO formula which does the following. For each segment vertex $v$ in $\patchwork$ whose crossings-label is $i$ (for $i\in \{0,\dots,\ell\}$), it identifies a set $L$ of \emph{all} segment vertices belonging to the edge corresponding to $v$. Then, for each such individually identified set $L$, the formula requires that $L$ intersects its interpretation in at most $\ell-i$ distinct vertices. This auxiliary formula is easy to construct and ensures that each edge in $\drawing$ will contain at most $\ell$-many crossings after the insertion of the solution identified by $\psi_\tau(\cdot)$, as required.
	
	For $\ell$-textsc{PEI} and also \textsc{Locally Crossing-Minimal Edge Insertion}, it suffices to follow precisely the steps used to solve \textsc{S$\ell$-PEI} and \SDE, and then remove the separate subformulas (denoted $\phi^i_{\tau'}(\cdot)$ in the proof of Theorem~\ref{thm:main}) that were used to ensure simplicity from the final constructed MSO formulas.
\end{proof}}

 	\section{Adding a Single Edge}
 	\label{sec:onedge}
 	\latertitle{\section{Extended Version of Section~\ref{sec:onedge}}}
	\both{In this section we present a single-exponential fixed-parameter algorithm for \SCEI\ parameterized by $\ell$ in the case where $|F|=1$; we hereinafter denote this problem \SConeEI.
	We remark that this algorithm is tight under the exponential time hypothesis~\cite{ImpagliazzoPZ01}, since Arroyo et al.~\cite{ArroyoKPSVW20} gave a reduction from 3-SAT to the simple drawing extension problem with $1$ extra edge, and the number of edges in the obtained graphs is linear in the size of the 3-SAT instance.
	}

	\short{
	As the first step towards our algorithm, we transform \SConeEI\
	to the problem of finding a colorful path (i.e., a path where no color is repeated) of length at most $\kappa$ in a vertex-colored graph obtained from $\planarization{\drawing}$.
	Finding this colorful path then amounts to a
	straightforward application of so-called \emph{representative sets}, see e.g.~\cite[Chapter~12]{CyganFKLMPPS15}.}

\later{
	As the first step towards our algorithm, we transform \SConeEI\
	to the problem of finding a colorful path of length at most $\kappa$ in the dual of the planarization of the given drawing.
	Finding this colorful path then amounts to a
	straightforward application of so-called \emph{representative sets}, see e.g.~\cite[Chapter~12]{CyganFKLMPPS15}.
	
	\begin{center}
		\begin{boxedminipage}{0.98 \columnwidth}
			\textsc{Colorful Short Path}\\[2pt]
			\begin{tabular}{l p{0.80 \columnwidth}}
				Input: & A graph $G$ with a vertex-coloring $\chi : V(G) \rightarrow [|V(G)|]$,
				two vertices $s,t \in V(G)$, and a positive integer $\kappa$.\\
				Question: & Does $G$ have a colorful $s$-$t$-path (i.e., a path where no color occurs more than once) of length at most $\kappa$?
			\end{tabular}
		\end{boxedminipage}
	\end{center}}
	
\both{\begin{proposition}[$\star$]		
	\label{prop:oneedgetopaths}
	There is a linear-time reduction that converts an instance $(G, \drawing, \{st\}, \ell)$ of \SConeEI\ to an equivalent instance $(G,\chi,s,t,2\ell+3)$ of \textsc{Colorful Short Path}.
	\end{proposition}}
	
\later{
	\begin{proof}
	Consider an instance $(G,\drawing,\{st\},\ell)$ of \SConeEI.
	We transform this instance into an instance of \textsc{Colorful Short Path} as follows.
	
	Let $G^*_0$ be the plane dual of the planarization of \drawing, and let $G^*$ be obtained from $G^*_0$ by subdividing every edge. In particular, $G^*$ is a bipartite graph where every vertex corresponds to either a cell (i.e., it is a cell vertex) or an edge segment (in which case we call it a segment vertex) in \drawing, and adjacencies represent incidences between cells and edge segments in \drawing. Moreover, we add $s$ and $t$ to $G^*$ as special marker vertices, and for each cell $c$ incident to $s$ (or to $t$) in $\drawing$, we add the edge $sc$ ($ct$) to $G^*$. Next, we assign to each edge in \drawing a unique color $i \in [|E(G)|]$. For the coloring function $\chi$, we use assign a unique color to each cell vertex of $G^*$ as well as to $s$ and $t$. Finally, for each segment vertex $e$, we color it using the color of its primal edge in \drawing.
	
	
		To conclude the proof, notice that $(G,\drawing,\{st\},\ell)$ is a \texttt{yes}-instance of \SConeEI\ if and only if $G^*$ contains a colorful $s$-$t$ path of length $2\ell+3$.
	Indeed, every drawing of the edge $st$
	with endpoints $s$ and $t$ in $\drawing$ allows us to construct such a colorful path---simply use the cells visited by that curve and the edge segments it intersects; colorful-ness follows from the fact that the curve cannot cross an edge in $\drawing$ more than once.
	Similarly, every colorful $s$-$t$-path allows us to obtain a solution to $(G,\drawing,\{st\},\ell)$ by having the solution curve intersect the cells and edge segments specified by the $s$-$t$ path, in the same order in which they are visited by the path.
	\end{proof}

	Our algorithm for \textsc{Colorful Short Path} uses representative sets in a manner that is similar to the presentation provided in, e.g., the book by Cygan et al.~\cite{CyganFKLMPPS15}. Intuitively, for each vertex $u$ in an instance of \textsc{Colorful Short Path}, the algorithm will dynamically compute a family of color sets, where the colors used by each $s$-$u$ path will form one set in the family. The caveat is that this family may become too large to effectively compute and store. Once a family for some vertex $u$ becomes larger than a specified bound (depending only on $\kappa$), representative sets allow us to prune some sets from our family while maintaining the property we need---in particular, if there was a way to extend the original family to a solution, then there is a way to extend the pruned family to a solution as well.
	
	Formally, one can define $q$-representative sets on matroids.
	Let $\mathcal M = (E,\mathcal I)$ be a matroid over universe $U$.
	We say that $A \subseteq U$ \emph{fits} $B\subseteq U$ if $A \cap B = \emptyset$ and $A\cup B$ is independent for $\mathcal M$.
	
	\begin{definition}[Definition~12.14 in \cite{CyganFKLMPPS15}]
		Let $\mathcal M$ be a matroid and $\mathcal A$ be a family of sets of size $p$ in $\mathcal M$.
		A subfamily $\mathcal A' \subseteq \mathcal A$ is said to $q$-represent $\mathcal A$ if for every set $B$ of size $q$ such that
		there is an $A \in \mathcal A$ that fits $B$, there is an $A' \in \mathcal A'$ that also fits $B$.
		If $\mathcal A'$ $q$-represents $\mathcal A$, we write $\mathcal A'  \qrep \mathcal A$
	\end{definition}
	
	While representative sets can be computed even for more general matroids,
	we only require uniform matroids of bounded rank.
	The following theorem allows us to compute a $q$-representative family
	of a family of $p$ sets in $\FPT$ time in $p$ and $q$.
	The following theorem was independently proven by Fomin et al.~\cite{DBLP:journals/jacm/FominLPS16} and
	Shachnai and Zehavi~\cite{DBLP:journals/corr/ShachnaiZ14}.
	
	\begin{theorem}[Theorem~12.31 in \cite{CyganFKLMPPS15}, see also \cite{DBLP:journals/jacm/FominLPS16,DBLP:journals/corr/ShachnaiZ14}]
		\label{thm:qrep}
		Let $\mathcal M$ be a uniform matroid over universe $U$ and $\mathcal A$ be a $p$-family of independent sets of $\mathcal M$.
		There is an algorithm that given $\mathcal A$, a rational number $0 < x < 1$,  and an integer $q$ 
		computes a $q$-representative family $\mathcal A' \qrep \mathcal A$
		of size at most $x^{-p}(1-x)^{-q}2^{o(p+q)}$ in time $|\mathcal A|(1-x)^{-q}2^{o(p+q)}\log|U|$.
	\end{theorem}

	A useful property of representative sets is that they are transitive.
	\begin{lemma}[Lemma~12.27 in \cite{CyganFKLMPPS15}]
		\label{lem:qreptrans}
		If $\widehat{\mathcal A}$ $q$-represents $\widetilde{\mathcal A}$ and $\widetilde{\mathcal A}$ $q$-represents $\mathcal A$, 
		then $\widehat{\mathcal A}$ $q$-represents $\mathcal A$.
	\end{lemma}

	Furthermore, the union of representative sets is again a representative set.
	\begin{lemma}[Lemma~12.26 in \cite{CyganFKLMPPS15}]
		\label{lem:qrepunion}
		If $\mathcal A_1$ and $\mathcal A_2$ are both $p$-families, $\mathcal A_1'$ $q$-represents $\mathcal A_1$ and
		$\mathcal A_2'$ $q$-represents $\mathcal A_2$, then $\mathcal A_1' \cup \mathcal A_2'$ $q$-represents $\mathcal A_1\cup \mathcal A_2$.
	\end{lemma}

	To build the representative sets we make use of the operation of set convolution to iteratively construct the representative families.
	Let $\mathcal A$ and $\mathcal B$ be two families, with $\mathcal A * \mathcal B$ 
	we denote the \emph{convolution} of $\mathcal A$ and $\mathcal B$, i.e.
	\begin{align*}
		\mathcal A * \mathcal B = \{A \cup B \mid A \in \mathcal A, B \in \mathcal B, \emph{and } A \cap B = \emptyset\}.
	\end{align*}

	\begin{lemma}[Lemma~12.28 in \cite{CyganFKLMPPS15}]
		\label{lem:qrepconv}
		Let $\mathcal A_1$ be a $p_1$-family and $\mathcal A_2$ be a $p_2$-family. Suppose $\mathcal A_1'$ $(k-p_1)$-represents $\mathcal A_1$ and $\mathcal A_2'$ $(k-p_2)$-represents $\mathcal A_2$. Then $\mathcal A_1'*\mathcal A_2'$ $(k-p_1-p_2)$-represents $\mathcal A_1 * \mathcal A_2$.
	\end{lemma}

	The proof of the following theorem closely follows the ones given by Cygan et al.~\cite[Chapter~12]{CyganFKLMPPS15}
	for longest cycle and path.}
	
	\both{\begin{theorem}[$\star$]
		\label{thm:colorfulpath}
		\textsc{Colorful Short Path} can be solved in time $\bigoh(2^{\bigoh(\kappa)}\cdot |E(G)|\log |V(G)|)$.
	\end{theorem}}
\short{
	\begin{proof}[Proof Sketch]
	Let $\mathcal{I}=(G,s,t,\chi,\kappa)$ be an instance of \textsc{Colorful Short Path}. Let $\mathcal P_u^p$ be the family of sets defined as follows: for $X\subseteq [|V(G)|]$, we have $X\in \mathcal P_{u}^p$ if and only if there is a colorful $s$-$u$ path $W$ in $G$ of length $p$ such that all colors in $X$ are used along $W$. Clearly, $\mathcal P_t^z\neq \emptyset$ for some $z\in [\kappa]$ if and only if $\mathcal{I}$ is a \texttt{yes}-instance. We now proceed by dynamic programming, where at each step $i\in [\kappa]$ we compute the sets $\mathcal P_v^i$ for each vertex $v$ of $G$. Whenever a computed set, say $Q$, becomes too large to effectively store and process by our fixed-parameter algorithm, we invoke the representative set machinery to replace $Q$ with a representative $Q'$ that can be used instead of $Q$ for our purposes.
	\end{proof}}
	
	\later{
	\begin{proof}
		Let $(G,s,t,\chi,\kappa)$ be an instance of \textsc{Colorful Short Path} with an $n$-vertex graph $G$.
		We consider the uniform matroid of rank at most $p + q$ with universe $V(G) \cup [n]$.
		Let $\mathcal P_u^p$ be the family of sets defined as follows. For $X\subseteq [n]$, we have $X\in \mathcal P_{u}^p$ if and only if there is a colorful $s$-$u$ path $W$ in $G$ of length $p$ such that all colors in $X$ are used along $W$.
		
		Clearly, if $\mathcal P_t^\kappa$ is non-empty there exists a colorful path of length at most $\kappa$ that starts at $s$ and ends at $t$.
		We will check this property by computing a family $\widehat{\mathcal P}_t^{\kappa,0} \xrep{0} \mathcal{P}_t^z$ and
		then test if $\widehat{\mathcal P}_t^{\kappa,0}$ is non-empty.
		This is sufficient for the following reason.
		Suppose $\mathcal P_t^\kappa$ is not empty.
		Consequently it contains a set that fits with the empty set $\emptyset$,
		but then $\widehat{\mathcal P}_t^{\kappa,0} \xrep{0} \mathcal{P}_t^\kappa$ also must contain such a set that fits with $\emptyset$ as 
		it is a $0$-representative set of $\mathcal P_t^\kappa$.
		
		To compute the representative sets we cannot just invoke Theorem~\ref{thm:qrep} as
		it would require us to enumerate all members of $\mathcal P_t^\kappa$.
		Instead, we are going to construct the representative sets for every vertex in an iterative fashion, with the aim of obtaining $\widehat{\mathcal P}_t^{\kappa,0}$ without directly computing $\mathcal P_t^\kappa$.
		Note that we will compute a representative set only for every second value of $p$,
		because we will always be adding a color and a vertex simultaneously in one step.
		In the beginning let us set
		$\widehat{\mathcal P}_s^0 = \mathcal P_s^0 = \{\{\chi(s)\}\}$ and
		$\widehat{\mathcal P}_u^0 = \mathcal P_u^0 = \emptyset$ for every $u \in V(G) \setminus \{s\}$.
		We proceed in rounds, iterating over the values in $\{1,\ldots,\kappa\}$ in increasing order.
		At the start of each loop we will maintain the invariant that for every $u \in V(G)$, $1 \leq p'\ \leq p$, and $q \leq \kappa - p'$
		we have computed a family $\widehat{\mathcal P}_u^{p',q}$ of size at most 
		\begin{align*}
			\left(\frac{p'+2q}{p'}\right)^{p'}\left(\frac{p'+2q}{2q}\right)^{q} \cdot 2^{o(p' + q)}
		\end{align*}
		 that $q$-represents $\mathcal P_u^{i}$.
		
		Now, consider the $p$-th iteration of our loop.
		We compute a new family $\widehat{\mathcal P}_u^{p,q}$ for each $v \in V(G)\setminus\{s\}$ and $q \leq \kappa - p$ as follows
		\begin{align*}
			\widetilde{\mathcal P}_v^{p, q} = \bigcup_{uv \in E(G)} \widehat{\mathcal P}_u^{p - 1, q + 1} * \{\{\chi(v)\}\}.
		\end{align*}
%
		Together, Lemma~\ref{lem:qrepunion} and~\ref{lem:qrepconv} imply that 
		$\widetilde{\mathcal{P}}_v^{p,q}$ $q$-represents $\mathcal P_v^p$.
		At this point we invoke Theorem~\ref{thm:qrep} to compute $\widehat{\mathcal P}_v^{p,q} \qrep \widetilde{\mathcal{P}}_v^{p,q}$.
		By Lemma~\ref{lem:qreptrans} it follows that $\widehat{\mathcal P}_v^{p,q}$ in fact $q$-represents $\mathcal P_v^p$.
		Finally, check if $\widehat{\mathcal P}_t^{p,q}$ is non-empty for any $q \leq \kappa - p$ if so 
		we have found a colorful path from $s$ to $t$ of length $p$.
		Backtracking our decisions, the discovered set in $\widehat{\mathcal P}_t^{p,q}$ 
		can be turned into a colorful path from $s$ to $t$ of length $p \leq \kappa$.
		
%
%
%
		It remains to compute the running time and space of our algorithm and show that we in fact can retrieve the paths in the same time.
		We apply Theorem~\ref{thm:qrep} with $x_{p.q} = \frac{p}{p+2q}$.
		With $i \leq p$ and following the same analysis as Cygan et al.~\cite[Lemma~12.33]{CyganFKLMPPS15} we obtain that 
		the sizes of the computed families $\widehat{\mathcal P}_u^{p,q}$ 
		for every $u \in V(G)$, $2 \leq p \leq \kappa$, and $q \leq \kappa - p$ are bounded by
		\begin{align*}
			\left(\frac{p+2q}{p}\right)^p\left(\frac{p+2q}{2q}\right)^{q} \cdot 2^{o(p + q)}
		\end{align*}
		and that the running time is given by the maximum of the function
		\begin{align*}
			f(p,q) = \left(\frac{p+2q}{p}\right)^p\left(\frac{p+2q}{2q}\right)^{2q}.
		\end{align*}
		times $n^{O(1)}$ and $2^{o(p+q)}$.
		Maximizing $f$ over the domain $1 \leq p \leq \kappa$ and $0\leq q \leq \kappa-p$ gives us that the maximum is attained for
		$p = (1 - \frac{1}{\sqrt 5})$ and $q = \kappa - p$.
		Evaluating $f$ for these values leads to a runtime that is upper bounded by $2.619^{\kappa + o(\kappa)} \cdot |E(G)|\log|V(G)|$.
		
		Computing the path from a set $S \in \widehat{\mathcal P}_t^z$ can be done by keeping the families
		in order of when they are computed.
		For example we could store them in a matrix of size $\kappa \times n$ with each row 
		correlating to a step and the columns to the vertices in $V(G)$.
		Backtracking our decisions we can obtain the colorful path of length at most $\kappa$.
	\end{proof}}

	\both{Theorem~\ref{thm:addoneedge}
	is an immediate consequence of Proposition~\ref{prop:oneedgetopaths} 
	together with Theorem~\ref{thm:colorfulpath}.
	
	\begin{theorem}\label{thm:addoneedge}
		\SConeEI can be solved in time $\bigoh(2^{\bigoh(\ell)}\cdot |\drawing|\log |E(G)|)$.
	\end{theorem}}
	
	\section{Conclusion}
		In this paper we established the fixed-parameter tractability of inserting a given set of edges into a given drawing while maintaining simplicity and adhering to various restrictions on the number of crossings in the solution.
	While the presented results make the reasonable assumption that the initial drawing is connected, the problem is of course also interesting in the general case. We believe that our framework and methodology can also be used to handle the extension problem for disconnected drawings, albeit only after overcoming a few additional technical challenges; moreover, the algorithm presented in Section~\ref{sec:onedge} does not require connectivity at all.
	Other than connectivty, the most glaring question left open concerns the complexity of \SCEI\ parameterized by $\ell$ alone. 
	Last but not least, while here we focused on the edge insertion problem, it would also be interesting to extend the scope to 
	also allow for the addition of vertices into the drawing.

	\bibliographystyle{plainurl}
	\bibliography{add_few_edges}

	\clearpage
	\appendix
	
	\ifbool{long}{}{\magicappendix}
\end{document}